\documentclass[runningheads]{llncs}

\usepackage{geometry}
\usepackage{amsmath,amssymb}
\usepackage{enumitem}
\usepackage{hyperref}
\usepackage[ruled]{algorithm2e} 

\usepackage{graphicx}
\usepackage{url}
\usepackage{diagbox}
\usepackage[comma, numbers]{natbib}

\usepackage{bm}
\usepackage{zref-clever}
\zcsetup{nameinlink=false}
\newcommand{\Cref}[1]{\zcref[S]{#1}}
\newcommand{\cref}[1]{\zcref{#1}}

\AddToHook{env/theorem/begin}{\zcsetup{countertype={theorem=theorem}}}
\AddToHook{env/lemma/begin}{\zcsetup{countertype={theorem=lemma}}}
\AddToHook{env/proposition/begin}{\zcsetup{countertype={theorem=proposition}}}
\AddToHook{env/corollary/begin}{\zcsetup{countertype={theorem=corollary}}}
\AddToHook{env/claim/begin}{\zcsetup{countertype={theorem=claim}}}
\AddToHook{env/example/begin}{\zcsetup{countertype={theorem=example}}}
\zcRefTypeSetup{claim}{
Name-sg = Claim ,
name-sg = claim ,
Name-pl = Claims ,
name-pl = claims
}

\usepackage{tikz}
\usepackage{pgfplots}
\pgfplotsset{compat=1.18}

\usepackage{xcolor}

\DeclareMathOperator*{\E}{{\mathop{\mathbb{E}}}}
\DeclareMathOperator*{\I}{{\mathop{\mathbb{I}}}}

\newcommand{\realc}{{\underline{c}}}
\DeclareMathOperator{\conv}{\mathop{\mathrm{conv}}}

\title{Strategic Disclosure of Action Space in Principal-Agent Contracts}

\author{Xiaotie Deng \inst{1,2}
\and Ningyuan Li \inst{1}}
\institute{
Center on Frontiers of Computing Studies, School of Computer Science, Peking University\\
\email{liningyuan@pku.edu.cn}\and
Department of Computer Science, City University of Hong Kong\\
\email{csdeng@cityu.edu.hk}
}

\begin{document}
\maketitle

\begin{abstract}
We study strategic disclosure of the action space in principal–agent contracting, where an agent selects a disclosed action set to shape the principal's perception of her capabilities before contract design. Unaware of the strategic disclosure, the principal designs a revenue-optimal contract as if the disclosed action set were complete and accurate. We consider two variants distinguished by cost verifiability. When costs are unverifiable, the agent can extract the entire first-best surplus, leaving the principal with zero revenue. When costs are verifiable, we characterize the agent's optimal disclosure strategy in binary-outcome settings and, more generally, when the principal is restricted to linear contracts, reducing the agent's problem to a two-variable convex optimization problem. We prove that the agent can secure utility of at least a $1/e$ fraction of the first-best surplus, which also yields a $1/e$ welfare guarantee under optimal disclosure. While the principal's revenue can be arbitrarily small compared to the first-best surplus, when the ratio of maximum to minimum expected reward among non-null base actions is at most $L$, we establish a revenue guarantee of $\Theta(1/\log L)$ relative to the first-best surplus. We also compare utilities and welfare under strategic disclosure with their counterparts in the canonical model. Finally, we extend the agent's $1/e$ utility guarantee to general outcome spaces without restricting the principal to linear contracts. Our results show how strategic action-space disclosure changes the distribution of surplus while preserving a constant-factor welfare guarantee under the agent's optimal disclosure.

\end{abstract}
\section{Introduction}
\label{sec:introduction}
The principal-agent model with hidden actions provides a foundational framework for analyzing incentive design under moral hazard \cite{holmstrom1979,GrossmanHart1983}. A risk-neutral principal delegates a task to an agent whose unobservable action influences stochastic outcomes. 
A standard assumption in this canonical setting is that the principal knows the agent's entire feasible action space a priori, which enables the classical result that a revenue-optimal contract can be computed by solving linear programs. Although analytically convenient, this assumption often fails in real-world environments where the principal lacks complete knowledge of the agent's capabilities.
Recognizing this limitation,  a growing literature studies contract settings without full knowledge of available actions, such as learning optimal contracts through repeated interactions (e.g. \cite{HoSV14}), robust contract design under uncertainties (e.g. \cite{Carroll15AER,SimplevsOptimalDuttingRT19}), Bayesian contract design with private agent types (e.g. \cite{GuruganeshSW21,AlonDT21}).
In most of these approaches, the imperfect knowledge is an exogenous constraint for the principal: the principal is fully aware of her uncertainty regarding the agent's feasible actions, and proactively mitigates it through contract design. 

In reality, however, a principal can be unaware of her own lack of knowledge. In practice, a principal may infer the agent's action space from observational data in historical interactions, and uses it as her prior knowledge in contract design. This creates a vulnerability: A strategic agent can consistently choose actions to misrepresent her capabilities, shaping the principal's perceived action space over time, while the principal remains oblivious to both the manipulation and the incompleteness of her knowledge. For instance, a software developer might hide streamlined workflows during project scoping, or a sales representative might deliberately underperform to inflate perceived effort costs. By selectively showcasing skills or masking efficiencies, the agent endogenously determines the action space perceived by the principal. The principal designs the contract naively as in the canonical model, but based on this strategically disclosed action space.

In this paper, we formalize this phenomenon as \textit{strategic disclosure of action space} within the principal-agent model. While this phenomenon naturally arises in repeated contracting interactions, we abstract it into a three-stage model to isolate the pure effect of endogenous epistemic manipulation from dynamic frictions.
In the first stage, the agent privately knows the true feasible action space $F$, and selects a disclosed action set $A$. In the second stage, the principal, perceiving only $A$ and unaware of the strategic selection, designs a revenue-optimal contract based on $A$, anticipating the best-response action in $A$ as in the canonical model. In the third stage, the agent implements the best-response action anticipated by the principal, maintaining the principal's unawareness of manipulation. Note that the second and third stages represent the limit case of repeated interactions where the principal's beliefs have converged.

Critically, to credibly shape the principal's belief about the action set, the agent's behavior must be indistinguishable by the principal from that of an ``honest" agent who always plays the best-response action in action set $A$ when offered any contract. This imposes implementation constraints on the agent's selection of disclosed action set and the choice of action.
Specifically, the principal's power to verify the implementation of actions depends on what she can observe or verify in repeated interactions. We assume the principal can verify that the outcome distribution of a disclosed action is implemented by the agent, for example by monitoring the empirical distribution of outcomes over time, even though the played actions are hidden. On the other hand, the verifiability regarding action costs leads to two variants of the model. When costs are verifiable by the principal, every disclosed action must be exactly implemented at its stated cost.
When costs are unverifiable, the agent may associate disclosed actions with arbitrary stated costs, while actually implementing their outcome distribution at the minimum feasible costs.

\subsection{Main Results}
We analyze the agent's optimal strategy of strategic disclosure under the two variants.
Under unverifiable costs, we show that the agent extracts the entire first-best surplus, leaving the principal with zero revenue. The optimal strategic disclosure is a two-action set comprising a null action and a surplus-maximizing action with inflated cost.

Under verifiable costs, the agent's power of strategic disclosure is constrained but remains substantial.
Our results for this variant mainly focus on the binary-outcome setting, or equivalently, when the principal is restricted to linear contracts. We characterize the agent's optimal disclosure strategy by optimizing a disclosed cost curve and reduce the resulting problem to a two-variable convex optimization problem. This characterization yields structural properties that we use to analyze the agent's utility and the principal's revenue.

We prove that the agent's utility is at least a $1/e$ fraction of the first-best welfare, and this guarantee can be extended to general outcome space without restricting the principal to linear contracts. This result yields an immediate implication that strategic disclosure ensures approximately efficient social welfare.

The principal's revenue can be arbitrarily small relative to the first-best welfare in worst-case instances, but achieves a $\Theta(1/\log L)$ approximation when the ratio between the maximum and minimum reward among feasible non-null actions is at most $L$. This suggests that the approximate optimality of linear contracts persists under agent's strategic disclosure.

We also analyze the utility impacts of strategic disclosure  relative to the canonical model where the principal knows the true feasible action space. While the improvement of the agent's utility can be unbounded, the principal's revenue is never improved, and can become arbitrarily worse. The social welfare can be either improved or damaged.

Omitted proofs are provided in the appendix.

\subsubsection{Note added after acceptance.} This manuscript presents the results of our WINE 2026 submission. We have subsequently obtained results for the verifiable-cost setting with general outcomes via a different approach, which will be presented in a substantially revised manuscript.


\subsection{Related Work}
\label{subsec:related_work}
\textbf{Technology design.}
A closely related model is studied in \cite{Garrett2023OptimalTD}, where an agent can choose and commit to any production technology set to influence contract design. While motivated by different scenarios from ours, their model is mathematically equivalent to the special case of our verifiable-cost model without a feasibility constraint on action costs. They explicitly derive the agent's optimal choice, under which the agent's utility and the principal's utility each equal an $1/e$ fraction of the maximum action reward. In our paper, however, the agent's implementation constraint with a general feasible action set in our model leads to critical distinctions in results and analysis. As the agent's optimal strategy can only be implicitly characterized, different techniques are required to establish structural results and worst-case guarantees.

\textbf{Robust contract design with action uncertainty.}
The study of robust contract design under uncertainty about action set was pioneered by \cite{Carroll15AER} and extended in recent works \cite{KAMBHAMPATI2023105585,kambhampati2024,PengT24}. In this framework, the principal is aware of the exogenous uncertainty regarding the action set, and optimizes the contract for worst-case robustness. In our model, while the principal also lacks full knowledge of the agent's feasible actions, the principal is unaware of this ignorance, and designs a revenue-maximizing contract naively based on her (misinformed) belief. This distinction shifts the analytical focus from robustness of principal's contract design to strategic outcomes under agent's endogenous action set disclosure.

\textbf{Learning optimal contracts.}
Our model is motivated by the agents' potential manipulation power in settings where principals infer agent capabilities through repeated interactions.
Starting from \cite{HoSV14}, the problem of learning optimal contracts in repeated contracting with agents has been extensively studied \cite{DBLP:conf/sigecom/ZhuBYWJJ23,DBLP:conf/iclr/BacchiocchiC0024,DBLP:conf/sigecom/00020DH24,DBLP:conf/sigecom/GuruganeshSW023,DBLP:journals/tcs/CohenDK23}. Most existing approaches assume either that a new agent arrives each round or that a single agent responds myopically to offered contracts, thereby neglecting long-term incentives. Our model captures the power of an agent with long-term incentives to manipulate a learning principal who is unaware of the potential distortion.


\textbf{Information design and information disclosure.} 
Information design studies a sender's design of a signaling scheme which sends signals to influence a Bayesian receiver's decisions \cite{KamenicaG11,BERGEMANN2007580,doi:10.1086/657922,aeri20210504}.
Recent works have explored the application of information design in contract settings. \cite{CastiglioniC25} study a contract setting where the principal reveals information about a hidden state to the agent by a signaling scheme, jointly designed with the contract. \cite{BabichenkoTXZ26} study a third party's design of the principal's information structure monitoring the agent's action. 
Although our model involves the agent's strategic disclosure of action set, the principal does not make Bayesian updates as she is unaware of the manipulation, resulting in the non-Bayesian structure of our model, which significantly differs from the Bayesian signaling structure in existing literature.

\section{Model and Preliminaries}
\label{sec:model}
Our model abstracts the process of repeated interactions into a three-stage model, which extends the canonical hidden-action principal-agent model by introducing strategic disclosure of action space, where the agent shapes the principal's perception of her feasible action space to influence contract design. The model features a risk-neutral principal and a risk-neutral agent with limited liability. Outcomes are perfectly observable and contractible, while actions remain unobservable.

\subsection{Environment and True Feasible Actions}
Let $\mathcal{O}$ denote the space of observable outcomes. Each outcome $o \in \mathcal{O}$ yields reward $r(o)$ to the principal. Without loss of generality, assume there is a null outcome $o_0$ with $r(o_0)=0$, and that the rewards are normalized such that $r(o)\in[0,1]$. We assume $\mathcal{O}$ and $r(\cdot)$ are public knowledge.

An action $a=(q_a,c_a)$ of the agent is specified by an outcome distribution $q_a \in \Delta(\mathcal{O})$ and an effort cost $c_a \in[0,+\infty)$. For convenience, we denote $r(a):=\E_{o\sim q_a}[r(o)]$, i.e., the expected reward of an action $a$.
The agent has a closed set of \textit{base actions} $F_B\subseteq \Delta(\mathcal{O})\times[0,+\infty)$, which is privately known to the agent only. 
Assume there is a null action $a_0=(q_{a_0},c_{a_0})\in F_B$ such that $r(a_0)=\E_{o\sim q_{a_0}}[r(o)]=0$ and $c_{a_0}=0$.

We assume the agent has two primary capabilities:
\begin{itemize}
\item \textbf{Randomization over base actions.} The agent can play any convex combination of base actions. The principal cannot distinguish such mixtures from a single ``real" action, since actions are hidden and only the realized outcomes are observed.
\item \textbf{Arbitrary wasteful costs.} The agent can incur arbitrary additional costs without altering the outcome distribution.
\end{itemize}
The resulting \textit{feasible action space} $F$ is the upward-closure (in costs) of the convex hull of base action set $F_B$. Formally, let $\conv(F_B)$ denote the convex hull of $F_B$, and we define
\[
F:=\bigcup_{a=(q_a,c_a)\in \conv(F_B)}\left\{(q_a,c'):c'\geq c_a\right\}.
\]

We also denote the set of implementable outcome distributions by $F_q:=\{q\in\Delta(\mathcal{O}):\exists c\in[0,+\infty), (q,c)\in \conv(F_B)\}$.
For any outcome distribution $q\in\Delta(\mathcal{O})$, we denote the minimum feasible cost to implement $q$ as
\[\realc(q):=\begin{cases}
    \min\{c\geq 0:(q,c)\in F\},&q\in F_q,\\
    +\infty,&q\notin F_q.
\end{cases}\]

\subsection{Disclosure, Contracting, and Implementation}
The model consists of three stages.
In the first stage, the agent strategically decides a \emph{disclosed action set} $A \subseteq \Delta(\mathcal{O}) \times [0, \infty)$, which is non-empty and closed. 
In the second stage, upon observing $A$, the principal offers a contract $w: \mathcal{O} \to [0, \infty)$ satisfying limited liability ($w(o) \geq 0, \forall o \in \mathcal{O}$). 
In the third stage, the agent plays an action $a\in F$. To shape the principal's belief, the agent's selection of disclosed action set and action are subject to implementation constraints introduced later.

To shape the principal's belief about action set $A$, the agent must make the principal's observations indistinguishable from those generated by an ``honest" agent whose true action space is $A$ and always plays the best response action in $A$.
Therefore, when a disclosed action $\hat{a}\in A$ is anticipated by the principal, the agent must implement it by a feasible action $a\in F$ subject to the following constraints: 
\begin{itemize}
    \item We assume outcome distributions are verifiable by the principal, as the empirical outcome distributions are observed during repeated interactions. Thus, the anticipated outcome distribution $q_{\hat{a}}$ must be implemented exactly as $q_{a}=q_{\hat{a}}$.
    \item If costs are verifiable by the principal, the anticipated cost $c_{\hat{a}}$ must be implemented exactly as $c_a=c_{\hat{a}}$. Otherwise, the cost of the actually executed action may differ from the anticipated cost $c_{\hat{a}}$. This leads to two variant models detailed later.
\end{itemize}
Importantly, the implementation constraints apply to every disclosed action $a\in A$, including those not anticipated under the principal's optimal contract. This is because  to credibly shape the principal's belief that a specific action $a$ exists, the agent must be capable of consistently implementing it during interactions.

Believing the action set $A$, the principal \textit{anticipates} that the agent will play a best-response action in $A$:
\[
a^{\mathrm{ant}} \in \arg\max_{a \in A} \tilde{U}(a; w), \quad \text{where } \tilde{U}(a; w) := \mathbb{E}_{o \sim q_a}[w(o)] - c_a.
\]
Ties are broken in favor of the principal. 
Similar to the canonical principal-agent problem, the principal designs the contract to maximize her expected revenue anticipating this response. Specifically, her revenue from offering contract $w$ is
\begin{align*}
\mathrm{Rev}(w;A):=&\max_{a^{\mathrm{ant}}\in A}\E_{o\sim q_{a^{\mathrm{ant}}}}[r(o)-w(o)]\\
&\text{s.t. }a^{\mathrm{ant}}\in\arg\max_{a\in A}\tilde{U}(a;w).
\end{align*}
And the principal chooses a revenue-optimal contract $w^*\in\arg\max_{w}\mathrm{Rev}(w;A)$. Ties are broken in favor of the agent. 


\subsection{Two Variants of Cost Verifiability}
We now formalize how cost verifiability constrains the agent's actual implementation of an anticipated action. We consider two variants.

    \newcommand{\Uver}{U^{\mathrm{ver}}}
    \newcommand{\Uunv}{U^{\mathrm{unv}}}
\begin{itemize}
\item \textbf{Verifiable costs.} The costs of disclosed actions can be audited by the principal, although they remain noncontractible. Consequently, if the principal anticipates a disclosed action $a\in A$, the agent must exactly execute it. 
This leads to the constraint that the disclosed action set is a subset of the feasible action space:
\[A \subseteq F.\]
When the principal designs contract $w$ and anticipates action $a\in A$, the agent's actual utility is consistent with the anticipated utility, denoted as
\[
U^{\mathrm{ver}}(a,w):=\E_{o\sim q_a}[w(o)] - c_a=\tilde{U}(a;w).
\]
The agent's problem is to strategically select the disclosed action set $A$ to maximize her final utility. With a slight abuse of notation, we denote the agent's final utility induced by disclosed action set $A$ as
\begin{align*}
    \Uver(A):=&\max_{w^*,a^{\mathrm{ant}}\in A}\Uver(a^{\mathrm{ant}},w^*),\\
   \text{s.t. }&w^*\in\arg\max_{w}\mathrm{Rev}(w;A),\\
    &a^{\mathrm{ant}}\in\arg\max_{a\in A}\tilde{U}(a;w^*).
\end{align*}
The agent's problem under verifiable costs can be written as
\begin{align}
\max_{A\subseteq F}\Uver(A).\label{eq:problem-agentUver}
\end{align}

\item \textbf{Unverifiable costs.} The principal cannot verify the agent's actual cost.
Consequently, the agent may associate any non-negative cost $c_a\geq 0$ with an implementable outcome distribution $q_a\in F_q$. When the principal anticipates an action $a=(q_a,c_a)\in A$, the agent is only constrained to exactly deliver the outcome distribution $q_a$. Since costs are unverifiable, she implements $q_a$ at the minimum feasible cost $\realc(q_a)$, regardless of the cost $c_a$ claimed in the disclosure. While $c_a$ strategically influences the principal's contract design, it does not represent the agent's true implementation cost.

Thus, the agent may disclose any action set $A$ with feasible outcome distributions and arbitrary non-negative costs:
\[
A \subseteq F_q\times [0,+\infty).
\]
When the principal designs contract $w$ and anticipates action $a\in A$, the agent's actual utility is
\[
\Uunv(a,w):=\E_{o\sim q_a}[w(o)]-\realc(q_a).
\]
Note that the actual utility $\Uunv(a,w)$ can be different from the anticipated utility $\tilde{U}(a,w)$. The agent's anticipated best-response action maximizes the anticipated utility $\tilde{U}(a,w)$ among disclosed actions, while her choice of the disclosed action set $A$ aims to maximize the actual utility $\Uunv(a,w)$. 

With a slight abuse of notation, we denote the agent's actual utility induced by disclosed action set $A$ as
\begin{align*}
    \Uunv(A):=&\max_{w^*,a^{\mathrm{ant}}\in A}\Uunv(a^{\mathrm{ant}},w^*),\\
    \text{s.t. }&w^*\in\arg\max_{w}\mathrm{Rev}(w;A),\nonumber\\
    &a^{\mathrm{ant}}\in\arg\max_{a\in A}\tilde{U}(a;w^*).\nonumber
\end{align*}
The agent's problem under unverifiable costs is
\begin{align}
    \max_{A\subseteq F_q\times [0,+\infty)}\Uunv(A)\label{eq:problem-agentUunv}.
\end{align}
    
\end{itemize}


\section{Warm-Up: Unverifiable Costs}
\label{sec:unverifiable-cost}
In this section, we analyze the variant model with unverifiable costs, where the agent is able to disclose an action claiming an arbitrary cost but implement its outcome distribution at the minimal feasible cost.

We show that the agent is able to implement the surplus-maximizing action, while extracting the full surplus. 
We denote the maximum surplus among feasible actions as
\[
S^*:=\max_{a \in F} \bigl( r(a) - c_a \bigr).
\]

\begin{theorem}
\label{thm:unverifiable_full_surplus}
Under unverifiable costs, the agent's optimal utility satisfies
\[
\max_{A \subseteq F_q \times [0, +\infty)} U^{\mathrm{unv}}(A) = S^*.
\]
Specifically, let $
a^+ \in \arg\max_{a \in F}( r(a) - c_a )
$
be a surplus-maximizing action in the true feasible set $F$, there exists a disclosed action set $\hat{A} \subseteq F_q \times [0, +\infty)$ such that:
\begin{enumerate}
    \item The principal's optimal contract incentivizes an anticipated action $\hat{a} \in \hat{A}$ with $q_{\hat{a}} = q_{a^+}$;
    \item The agent implements the anticipated outcome distribution $q_{\hat{a}}$ by playing the feasible action $a^+$, bearing cost $c_{a^+}$;
    \item The agent receives expected payment of $r(a^+)$, extracting the entire surplus $S^*=r(a^+)-c_{a^+}$.
\end{enumerate}
\end{theorem}

\begin{proof}
Assume without loss of generality that $S^*=r(a^+)-c_{a^+}>0$, i.e., the maximal surplus is positive. (Otherwise $S^*=0$, and the agent obtains trivially optimal utility $0$ by selecting the action set $\hat{A}=\{a_0\}$, where $a_0$ is the null action.)

Define action $\hat{a}=(q_{\hat{a}},c_{\hat{a}})\in F_q\times[0,+\infty)$, such that $q_{\hat{a}}=q_{a^+}$ and $c_{\hat{a}}=r(a^+)$.
Under unverifiable costs, the outcome distribution of action $\hat{a}$ can be implemented by $a^+$.
Construct disclosed action set $\hat{A}=\{a_0,\hat{a}\}$, where $a_0$ is the null action.

We show that the principal's optimal revenue is $0$, obtained by incentivizing either action in $\hat{A}$.
One can easily see that the optimal revenue from any contract incentivizing $a_0$ is $0$.
For any contract incentivizing $\hat{a}$, the expected payment is at least $c_{\hat{a}}-c_{a_0}=r(a^+)$.
Since $r(a^+)\geq S^*>0$ and $r(a_0)=0$, there exists some outcome $\hat{o}\in\mathcal{O}$ with $r(\hat{o})>0$, such that $\hat{o}$ is in the support of $q_{a^+}$, but is not in the support of $q_{a_0}$. Therefore, the minimal expected payment to incentivize $\hat{a}$ equals $c_{\hat{a}}-c_{a_0}=r(a^+)$, e.g., by the contract $w(o)=\I[o=\hat{o}]\frac{r(a^+)}{q_{a^+}(\hat{o})}$. Thus, the principal's optimal revenue when incentivizing $\hat{a}$ is also $0$.

Under the assumption that the principal breaks ties in favor of the agent, the principal selects an optimal contract $w^*$ such that $\hat{a}$ is incentivized with expected payment $r(a^+)$.
However, it suffices for the agent to implement the outcome distribution of $\hat{a}$ at the minimum feasible cost by playing $a^+$, obtaining an actual utility of
\[\Uunv(\hat{A})=\Uunv(\hat{a},w^*)=\E_{o\sim q_{\hat{a}}}[w^*(o)]-\realc(q_{\hat{a}})=r(a^+)-c_{a^+}.\]
Here $\realc(q_{a^+})=c_{a^+}$ since $a^+$ achieves the maximal surplus.

Lastly, we verify the optimality of $\Uunv(\hat{A})$. Since the principal's utility is always non-negative, for any disclosed action set $A'$, suppose the principal's optimal contract incentivizes $a'\in A'$, it holds that
$\Uunv(A')\leq r(a')-\realc(q_{a'})$, i.e., the agent's utility is upper-bounded by the social surplus.
It follows that $\max_{A\subseteq F_q\times[0,+\infty)}\Uunv(A)\leq\max_{a\in F}r(a)-c_a=S^*$.
Therefore, we have
\[
\max_{A\subseteq F_q\times[0,+\infty)}\Uunv(A)=r(a^+)-c_{a^+}=S^*.
\]
\qed
\end{proof}
\Cref{thm:unverifiable_full_surplus} shows that the agent extracts full surplus while the principal obtains zero revenue under unverifiable costs. 

The optimal action set is constructed as $\hat{A} = \{a_0, \hat{a}\}$, where $a_0$ is a null action, and $\hat{a}$ has the same outcome distribution as $a^+$ while assigned a highly inflated cost. In this way, the agent shapes the principal's belief that no non-null action can be incentivized unless a high-wage contract is offered. The wage can be increased until zero revenue remains for the principal.



\textbf{Implications for Learning in Repeated Interactions.}
Our model serves as a limit case for a principal's learning process of the revenue-optimal contract through repeated interactions with a single agent. In such settings, the principal typically assumes that the agent is myopic, i.e., always taking a best response action to the offered contract at each round. Our result suggests that a forward-looking agent can significantly benefit from hiding their actual cost efficiency, especially when costs remain unverifiable. By consistently performing as if their costs are inflated during the interaction process, the agent manipulates the principal's eventual belief about the feasible action space, potentially increasing the payment received in the future.




\section{Verifiable Costs with Binary Outcomes}
\label{sec:verifiable-cost}
In this section and the next section, we analyze the variant model with verifiable costs, where the agent is constrained to exactly implement the outcome distribution and cost of any disclosed action $a\in A$ when it is anticipated.
In this section, we characterize the agent's optimal strategy of disclosed action set $A$ under the binary-outcome setting, where the principal's optimal contract is simplified to a linear contract. Our results also hold directly under general outcome spaces when the principal is restricted to linear contracts.


\subsection{Binary-Outcome Assumption}
In this section, we assume a binary outcome space $\mathcal{O}=\{o_0,o_1\}$ with normalized rewards $r(o_0)=0$ and $r(o_1)=1$. 
Under this specification, any action $a = (q_a, c_a)$ is fully characterized by its success probability $y_a := q_a(o_1) = \mathbb{E}_{o \sim q_a}[r(o)] \in [0,1]$ and its cost $c_a$. With a slight abuse of notation, we represent an action as a pair $a = (y_a, c_a)$ where $y_a \in [0,1]$ denotes the expected reward and $c_a \geq 0$ the action cost.

This binary-outcome setting admits an important simplification: for any disclosed action set $A$, the principal's optimal contract always sets zero wage for the null outcome $o_0$ and pays a non-negative amount only upon observing $o_1$. Consequently, the optimal contract is linear and can be represented by a single parameter $\alpha \in [0, 1]$, where the agent receives expected payment $\alpha \cdot y_a$ for implementing action $a$. Formally, the linear contract $w_\alpha$ satisfies $w_\alpha(o_0) = 0$ and $w_\alpha(o_1) = \alpha$, yielding $\mathbb{E}_{o \sim q_a}[w_\alpha(o)] = \alpha\cdot y_a$.

Our analysis for this binary-outcome setting extends directly to environments with general outcome spaces when the principal is restricted to linear contracts (i.e., $w(o) = \alpha \cdot r(o)$ for some fixed $\alpha \geq 0$). In both cases, the contract space reduces to a one-dimensional choice of $\alpha$, enabling a tractable characterization of the agent's optimal disclosed action set.

\subsection{Cost Curve Representation}
We firstly show that the agent's selection of disclosed action set $A$ is equivalent to designing a cost curve $c:[0,1]\to[0,+\infty]$.

\newcommand{\lowconv}{\mathop{\mathrm{lowerConv}}}
Given a two-dimensional point set $\Gamma$, we define its lower convex hull as a function 
\[\lowconv[\Gamma](x):=\begin{cases}
    \inf\{z:(x,z)\in\conv(\Gamma)\},& \exists z, (x,z)\in\conv(\Gamma),\\
    +\infty, &\text{otherwise.}
\end{cases}\]

Under verifiable costs and binary outcome space, any disclosed action set $A$ can be represented as a two-dimensional point set $\{a=(y_a,c_a)\}$ of reward-cost pairs. Since the principal only uses linear contracts, only the points on the lower convex hull of $A$ may be incentivized, and all other points have no influence on contract design. Therefore, a disclosed action set $A$ can be represented by its lower convex hull, which can be written as a cost curve for action reward $y\in[0,1]$:
\[
c(y):=\lowconv[A](y).
\]

\newcommand{\UpY}{\bar{Y}}
With a slight abuse of notation, we similarly view the feasible base action set $F_B$ as a two-dimensional point set of reward-cost pairs.
The minimal feasible cost for any action reward $y\in[0,1]$ can be written as
\[\realc(y):=\lowconv[F_B](y).\]
Specifically, the existence of null action $a_0\in F_B$ implies $\realc(0)=0$, which further implies that $\realc(y)$ is non-decreasing. Denote the maximum feasible action reward by $\UpY:=\max_{a\in F_B}y_a$,
then $\realc(y)<+\infty$ if and only if $y\leq\UpY$, indicating the feasibility.

The agent's selection of disclosed action set $A$ can be viewed as designing a convex cost curve $c:[0,1]\to[0,+\infty]$, and the constraint on disclosed action set that $A\subseteq F$ is rewritten as
\[
c(y)\geq \realc(y),\quad \forall y\in[0,1].
\]
Given the cost curve $c(y)$ and any linear contract $w_\alpha$ offered by the principal, the anticipated best-response action $(y^*,c(y^*))$ is given by
\begin{align*}
    y^*\in\arg\max_{y\in[0,1]}\alpha y-c(y),
\end{align*}
Ties are broken in favor of the principal, i.e., the largest optimal $y^*$ is taken.

We further analyze the induced optimal contract given the cost curve $c(y)$.
Since $\alpha$ is non-negative, it is without loss of generality to assume that $c(y)$ is non-decreasing on $[0,1]$. (Otherwise, there exists some $\hat{y}\in[0,1]$ such that $c(y)$ is non-increasing on $[0,\hat{y}]$ and non-decreasing on $[\hat{y},1]$, and setting $c(y)=c(\hat{y})$ for all $y\in[0,\hat{y}]$ makes the cost curve non-decreasing without influencing contract design.)

Define $c'_-(y)$ as the left-derivative of $c$ at $y\in(0,1]$. Specifically, we define $c'_-(0)=0$, and $c'_-(y)=+\infty$ when $c(y)=+\infty$. Then $c'_-(y)$ is non-decreasing and left-continuous on $[0,1]$. For any linear contract $w_\alpha$ with $\alpha\in[0,1]$, the anticipated best-response action is given by
\[y^*=\sup\{y\in[0,1]:c'_-(y)\leq \alpha\}.\]
By the left-continuity of $c'_-(y)$, it holds that $c'_-(y^*)\leq \alpha$. 
Moreover, in the optimal contract, it holds without loss of generality that $\alpha=c'_-(y^*)$: When $\alpha>c'_-(y^*)$, lowering $\alpha$ to $c'_-(y^*)$ weakly increases the principal's revenue, while $y^*$ is still incentivized.

Therefore, the principal's problem becomes
\begin{align*}
    \max_{y^*\in[0,\UpY]}(1-c'_-(y^*))y^*,
\end{align*}
where the principal selects the optimal solution of $y^*$, taking the maximum one if there are multiple optimal solutions. The optimal contract sets $\alpha=c'_-(y^*)$ and anticipates action $(y^*,c(y^*))$. $y^*$ can be restricted in $[0,\UpY]$ due to feasibility.

Now we can rewrite the agent's problem \eqref{eq:problem-agentUver} as the following optimization over the cost curve:
\begin{align}
\max_{A\subseteq F}\Uver(A)=&\max_{c:[0,1]\to [0,+\infty],y^*\in[0,\UpY]}c'_-(y^*)y^*-c(y^*)\label{problem:cost-curve}\\
\text{s.t. }&(1-c'_-(y^*))y^*\geq (1-c'_-(y))y,&\forall y\in[0,y^*),\nonumber\\
&(1-c'_-(y^*))y^*>(1-c'_-(y))y,&\forall y\in(y^*,1],\nonumber\\
&c(y)\geq\realc(y),&\forall y\in[0,1],\nonumber\\
&c(y)\text{ is non-decreasing and convex on }[0,1].\nonumber
\end{align}
Here $y^*$ is an auxiliary variable representing the induced action reward. The first and second lines of constraints ensure that $w_{c'_-(y^*)}$ is the principal's optimal linear contract and $y^*$ is incentivized. The third line of constraints enforces the feasibility of implementation, i.e., $A\subseteq F$. The last constraint follows from our characterization above.

\subsection{Optimal Cost Curve}
We now characterize the agent's optimal disclosed action set by solving the cost curve optimization problem \eqref{problem:cost-curve}. Since the optimal solution depends on $\realc(\cdot)$, there is generally no closed-form solution. However, we can reduce the problem to a two-variable convex optimization, which leads to an additive FPTAS. 

Our characterization consists of two main steps. First, we fix a target contract parameter $\alpha^*\in[0,1]$ and an anticipated action reward $y^*\in[0,1]$, and characterize the optimal cost curve for the agent to induce $\alpha^*$ and $y^*$ with maximal utility, which admits a closed-form solution. This reduces the infinite-dimensional optimization to a two-variable problem over $(\alpha^*,y^*)$.
Second, we reparameterize the problem using the principal's revenue $R^*:=(1-\alpha^*)y^*$ instead of $\alpha^*$. This converts the objective into a jointly concave function in $(R^*,y^*)$, reducing the problem to a convex optimization.


\newcommand{\reals}{\bar{s}}



\subsubsection{Optimizing Cost Curve for Fixed Contract and Reward}

Given any contract parameter $\alpha^*\in[0,1]$ and an anticipated action reward $y^*\in[0,\UpY]$, we analyze the problem \eqref{problem:cost-curve} with the additional constraint that the agent's design of cost curve induces the principal's selection of contract $w_{\alpha^*}$, with incentivized action reward $y^*$. 
Given these conditions, the agent's utility is $\alpha^*y^*-c(y^*)$. Since the first term $\alpha^*y^*$ is fixed, the optimal cost curve minimizes the cost $c(y^*)$ under the constraints.

Observe that we can simply set $c(y)=+\infty$ for $y\in(y^*,1]$, i.e., let $y^*$ be the maximal reward among disclosed actions. Assuming this, we can focus on the optimization of $c(\cdot)$ on $y\in[0,y^*]$, which we write as the following problem:
\begin{align}
\mathrm{MinCost}(\alpha^*,y^*):=&\inf_{c(\cdot)}c(y^*)\label{problem:min-cost-for-fixed-outcomes}\\
\text{s.t. }
&c'_-(y^*)=\alpha^*,\nonumber\\
&(1-\alpha^*)y^*\geq (1-c'_-(y))y,&\forall y\in[0,y^*),\nonumber\\
&c(y)\geq\realc(y),&\forall y\in[0,y^*],\nonumber\\
&c(y)\text{ is non-decreasing and convex on }[0,y^*].\nonumber
\end{align}
Then we can rewrite the agent's problem \eqref{problem:cost-curve} as
\begin{align}
    \max_{A\subseteq F}\Uver(A)=\max_{\alpha^*\in[0,1],y^*\in[0,\UpY]}\alpha^*y^*-\mathrm{MinCost}(\alpha^*,y^*).\label{eq:Uver-with-MinCost}
\end{align}

Below we characterize the optimal solution of $\mathrm{MinCost}(\alpha^*,y^*)$, and obtain a closed-form solution.
\begin{theorem}
\label{thm:min-cost-optimal-curve}
For any $\alpha^*\in[0,1]$ and $y^*\in[0,\UpY]$, it holds that
    \begin{align*}
        \mathrm{MinCost}(\alpha^*,y^*)&=\max_{y\in[0,y^*]}(\realc(y)+\int_{y}^{y^*}(1-\frac{(1-\alpha^*)y^*}{t})dt)\\
        &=\max_{y\in[0,y^*]}(\realc(y)+y^*-y-(1-\alpha^*)y^*(\ln y^*-\ln y)).
    \end{align*}
Moreover, the infimum in problem \eqref{problem:min-cost-for-fixed-outcomes} is attainable.

\end{theorem}

\begin{proof}

Firstly, define $g_{\alpha^*,y^*}(y):=1-\frac{(1-\alpha^*)y^*}{y}$, we prove a lower bound that $\mathrm{MinCost}(\alpha^*,y^*)\geq \max_{y\in[0,y^*]}\realc(y)+\int_{y}^{y^*}g_{\alpha^*,y^*}(t)dt$.
Observe that the second constraint in problem \eqref{problem:min-cost-for-fixed-outcomes} implies a lower bound for the left-derivatives:
\[c'_-(y)\geq 1-\frac{(1-\alpha^*)y^*}{y}=g_{\alpha^*,y^*}(y),\]
for $y\in[0,y^*)$. 


For any valid solution $c(\cdot)$ of problem \eqref{problem:min-cost-for-fixed-outcomes}, for any $y\in(0,y^*)$, we have 
\[c(y^*)=c(y)+\int_y^{y^*}c'_-(t)dt\geq \realc(y)+\int_y^{y^*}g_{\alpha^*,y^*}(t)dt,\]
where $c(y)\geq \realc(y)$ holds by the third constraint in \eqref{problem:min-cost-for-fixed-outcomes}.

Since this holds for all $y\in[0,y^*]$ and all feasible solutions $c(\cdot)$, we obtain
\begin{align*}
    \mathrm{MinCost}(\alpha^*,y^*)&\geq \max_{y\in[0,y^*]}\realc(y)+\int_y^{y^*}g_{\alpha^*,y^*}(t)dt\\
    &=\max_{y\in[0,y^*]}\realc(y)+\int_{y}^{y^*}\Big(1-\frac{(1-\alpha^*)y^*}{t}\Big)dt\\
    &=\max_{y\in[0,y^*]}\realc(y)+y^*-y+(1-\alpha^*)y^*(\ln y-\ln y^*).
\end{align*}

Secondly, we prove that this lower bound is achievable, by constructing a valid cost curve $c(\cdot)$ with $c(y^*)=\max_{y\in[0,y^*]}\realc(y)+\int_y^{y^*}g_{\alpha^*,y^*}(t)dt$.


Pick $y^\dagger \in \arg\max_{y \in [0, y^*]} \realc(y) + \int_y^{y^*} g_{\alpha^*,y^*}(t)\,dt$, breaking ties arbitrarily.
We construct the cost curve on $y\in[0,y^*]$ as
\[
c(y):=\realc(y^\dagger)+\int_{y^\dagger}^{y}\max\{0,g_{\alpha^*,y^*}(t)\}dt.
\]

Note that $g_{\alpha^*,y^*}(y)$ is non-decreasing in $y$, and that $g_{\alpha^*,y^*}(y)\geq0$ if and only if $y\geq(1-\alpha^*)y^*$.
Observe that we have either $y^\dagger=y^*$, or $g_{\alpha^*,y^*}(y^\dagger)\geq 0$, because otherwise $\realc(y) + \int_y^{y^*} g_{\alpha^*,y^*}(t)\,dt$ is strictly increasing at $y=y^\dagger$, which contradicts the definition of $y^\dagger$. It follows that $y^\dagger\geq (1-\alpha^*)y^*$, and that $c(y^*)=\realc(y^\dagger)+\int_{y^\dagger}^{y^*}g_{\alpha^*,y^*}(t)dt$. That is, the lower bound is achieved by $c(\cdot)$.

It remains to verify that $c(\cdot)$ satisfies the four constraints in problem \eqref{problem:min-cost-for-fixed-outcomes}.

The first constraint $c'_-(y^*)=\alpha^*$ holds since we have $c'_-(y)=g_{\alpha^*,y^*}(y)$ for all $y\geq (1-\alpha^*)y^*$, so $c'_-(y^*)=g_{\alpha^*,y^*}(y^*)=\alpha^*$.

The second constraint $(1-\alpha^*)y^*\geq (1-c'_-(y))y,~\forall y\in[0,y^*)$ holds, as we have $c'_-(y)=\max\{0,g_{\alpha^*,y^*}(y)\}\geq g_{\alpha^*,y^*}(y)$ for all $y\in[0,y^*]$, and the second constraint is implied.

Next we verify the third constraint $c(y)\geq \realc(y),~\forall y\in[0,y^*]$.
For any $y\in[(1-\alpha^*)y^*,y^*]$, we have $c(y)=\realc(y^\dagger)+\int_{y^\dagger}^{y}g_{\alpha^*,y^*}(t)dt$. Observe that
\[c(y)-\realc(y)=\left(\realc(y^\dagger)+\int_{y^\dagger}^{y^*}g_{\alpha^*,y^*}(t)dt\right)-\left(\realc(y)+\int_{y}^{y^*}g_{\alpha^*,y^*}(t)dt\right)\geq 0,\]
where the inequality is by definition of $y^\dagger$. So $c(y)\geq\realc(y)$ for all $y\in[(1-\alpha^*)y^*,y^*]$. For any $y\in[0,(1-\alpha^*)y^*)$, we have
\[c(y)=c((1-\alpha^*)y^*)\geq \realc((1-\alpha^*)y^*)\geq \realc(y),\]
where the second inequality is because $\realc(y)$ is non-decreasing.

The fourth constraint that $c(y)$  is non-decreasing and convex on $[0,y^*]$ holds as $c'(y)=\max\{0,g_{\alpha^*,y^*}(y)\}$ is non-negative and non-decreasing.

In summary, $c(\cdot)$ is a valid solution for \eqref{problem:min-cost-for-fixed-outcomes}, and achieves $c(y^*)=\max_{y\in[0,y^*]}\realc(y)+\int_y^{y^*}g_{\alpha^*,y^*}(t)dt$. It follows that
\[\mathrm{MinCost}(\alpha^*,y^*)=\max_{y\in[0,y^*]}\realc(y)+\int_y^{y^*}g_{\alpha^*,y^*}(t)dt=\max_{y\in[0,y^*]}\realc(y)+\int_{y}^{y^*}\Big(1-\frac{(1-\alpha^*)y^*}{t}\Big)dt.\]
This completes the proof.
\qed

\end{proof}


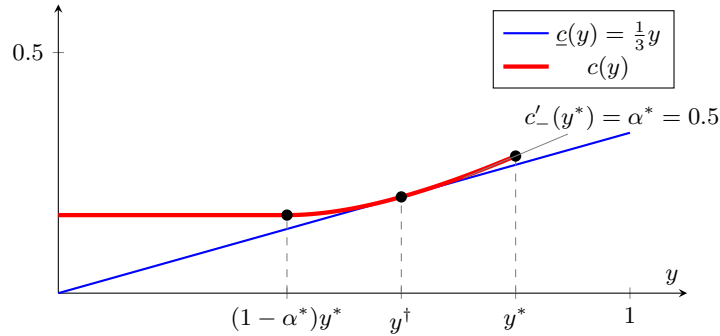
\begin{figure}[ht]
\centering
\begin{tikzpicture}
\begin{axis}[
    width=9.9cm,
    height=5.4cm,
    axis lines=middle,
    xlabel={$y$},
    ylabel={},
    xmin=0, xmax=1.1,
    ymin=0, ymax=0.6,
    domain=0:1,
    samples=200,
    clip=false, 
    grid=none,
    tick label style={font=\small},
    xtick={0, \yflat,\ydag, \ystar, 1},
    xticklabels={$0$,$(1-\alpha^*)y^*$, $y^\dagger$, $y^*$, $1$},
    ytick={0,0.5},
    yticklabels={$0$,$0.5$},
    label style={font=\small},
    legend style={font=\small, at={(0.98,0.98)}, anchor=north east}
]

\def\ystar{0.8}
\def\alphastar{0.5}
\def\slopeconst{0.4} 
\def\ydag{0.6}
\def\cdag{0.2}       
\def\yflat{0.4}      
\def\cflat{0.162186043243}      
\def\cstar{0.284927171019}      

\addplot [
    thick,
    blue,
    domain=0:1
] {1/3*x};
\addlegendentry{$\underline{c}(y) = \frac{1}{3} y$};

\addplot [
    ultra thick,
    red,
    domain=0:\yflat
] {\cflat};

\addplot [
    ultra thick,
    red,
    domain=\yflat:\ystar
] {\cdag + (x - \ydag) - \slopeconst*(ln(x) - ln(\ydag))};
\addlegendentry{$c(y)$};

\draw[dashed, gray] 
    (axis cs:\ystar,0) -- (axis cs:\ystar,\cstar) 
    node[circle, fill=black, inner sep=1.5pt]{};
\draw[dashed, gray] 
    (axis cs:\ydag,0) -- (axis cs:\ydag,\cdag) 
    node[circle, fill=black, inner sep=1.5pt]{};
\draw[dashed, gray] 
    (axis cs:\yflat,0) -- (axis cs:\yflat,\cflat) 
    node[circle, fill=black, inner sep=1.5pt]{};


\draw[help lines, shorten >=2pt, shorten <=2pt] 
    (axis cs:{\ystar-0.1},{\cstar-0.1*\alphastar}) -- 
    (axis cs:{\ystar+0.1},{\cstar+0.1*\alphastar});
\node[right, font=\small] at (axis cs:\ystar,\cstar+0.08) {$c'_-(y^*)=\alpha^*=0.5$};


\end{axis}
\end{tikzpicture}
\caption{The construction of optimal disclosed cost curve $c(y)$ for $\alpha^*=0.5$, $y^*=0.8$.}
\label{fig:cost_curve_example}
\end{figure}

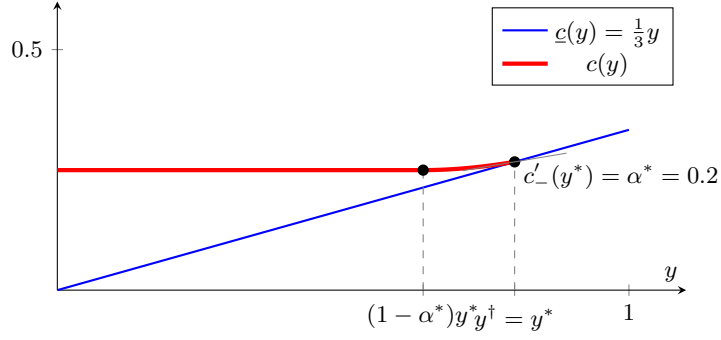
\begin{figure}[ht]
\centering
\begin{tikzpicture}
\begin{axis}[
    width=9.9cm,
    height=5.4cm,
    axis lines=middle,
    xlabel={$y$},
    ylabel={},
    xmin=0, xmax=1.1,
    ymin=0, ymax=0.6,
    domain=0:1,
    samples=200,
    clip=false, 
    grid=none,
    tick label style={font=\small},
    xtick={0,0, \yflat, \ystar, 1},
    xticklabels={$0$,$0$,$(1-\alpha^*)y^*$, $y^\dagger=y^*$, $1$},
    ytick={0,0,0.5},
    yticklabels={$0$,$0$,$0.5$},
    label style={font=\small},
    legend style={font=\small, at={(0.98,0.98)}, anchor=north east}
]

\def\ystar{0.8}
\def\alphastar{0.2}
\def\slopeconst{0.64} 
\def\ydag{0.8}
\def\cdag{0.266666666667}       
\def\yflat{0.64}      
\def\cflat{0.249478539508}      
\def\cstar{0.266666666667}      

\addplot [
    thick,
    blue,
    domain=0:1
] {1/3*x};
\addlegendentry{$\underline{c}(y) = \frac{1}{3} y$};

\addplot [
    ultra thick,
    red,
    domain=0:\yflat
] {\cflat};

\addplot [
    ultra thick,
    red,
    domain=\yflat:\ystar
] {\cdag + (x - \ydag) - \slopeconst*(ln(x) - ln(\ydag))};
\addlegendentry{$c(y)$};

\draw[dashed, gray] 
    (axis cs:\ystar,0) -- (axis cs:\ystar,\cstar) 
    node[circle, fill=black, inner sep=1.5pt]{};
\draw[dashed, gray] 
    (axis cs:\yflat,0) -- (axis cs:\yflat,\cflat) 
    node[circle, fill=black, inner sep=1.5pt]{};


\draw[help lines, shorten >=2pt, shorten <=2pt] 
    (axis cs:{\ystar-0.1},{\cstar-0.1*\alphastar}) -- 
    (axis cs:{\ystar+0.1},{\cstar+0.1*\alphastar});
\node[right, font=\small] at (axis cs:\ystar,\cstar-0.03) {$c'_-(y^*)=\alpha^*=0.2$};


\end{axis}
\end{tikzpicture}
\caption{The construction of optimal disclosed cost curve $c(y)$ for $\alpha^*=0.2$, $y^*=0.8$.}
\label{fig:cost_curve_example2}
\end{figure}
\begin{example}
Consider an example where the minimum feasible cost is $\realc(y)=\frac13 y$.
For $\alpha^*=0.5$ and $y^*=0.8$, according to \Cref{thm:min-cost-optimal-curve}, the optimal solution of $\mathrm{MinCost}(\alpha^*,y^*)$ is constructed as
$c(y)=\realc(y^\dagger)+\int_{y^\dagger}^{y}\max\left\{0,1-\frac{(1-\alpha^*)y^*}{t}\right\}dt$
for $y\in[0,y^*]$, where $y^\dagger=0.6$, as illustrated in \Cref{fig:cost_curve_example}. Recall that $c(y)=+\infty$ on the interval $(y^*,1]$, which is not illustrated.
The cost curve $c(y)$ is always on or above the minimum feasible cost curve $\realc(y)$, and the two curves are tangent at $y^\dagger=0.6$, where $c'(y^\dagger)=g_{\alpha^*,y^*}(y^\dagger)=1-\frac{(1-\alpha^*)y^*}{y^\dagger}=\frac13=\realc'(y^\dagger)$.

Consider another case with $\alpha^*=0.2$ and $y^*=0.8$. The constructed optimal solution of $\mathrm{MinCost}(\alpha^*,y^*)$ is illustrated in \Cref{fig:cost_curve_example2}. Recall that $c(y)=+\infty$ on the interval $(y^*,1]$, which is not illustrated. The difference from the former case is that the two curves are not tangent at the intersecting point $y^\dagger=y^*=0.8$.
\end{example}

\subsubsection{Reparameterization and Surplus Curve Representation}
By \Cref{thm:min-cost-optimal-curve} and equation \eqref{eq:Uver-with-MinCost}, now we have
\begin{align*}
    \max_{A\subseteq F}\Uver(A)&=\max_{\alpha^*\in[0,1],y^*\in[0,\UpY]}\alpha^*y^*-\mathrm{MinCost}(\alpha^*,y^*)\\
    &=\max_{\alpha^*\in[0,1],y^*\in[0,\UpY]}\alpha^*y^*-\max_{y\in[0,y^*]}(\realc(y)+y^*-y-(1-\alpha^*)y^*(\ln y^*-\ln y)).
\end{align*}

To simplify the optimization structure, we reparameterize the agent's problem with the principal's revenue 
\[R^*=(1-\alpha^*)y^*,\]
instead of the contract parameter $\alpha^*$.

We also define the feasible surplus curve
\[\reals(y)=y-\realc(y),\quad y\in[0,1],\]
which represents the social surplus obtained from implementing action reward $y$ with minimum feasible cost. We have $\reals(0)=0-\realc(0)=0$, and $\reals(y)$ is concave in $y\in[0,1]$ by the convexity of $\realc(y)$.

We rewrite the agent's optimization of utility as
\begin{align*}
&\max_{\alpha^*\in[0,1],y^*\in[0,\UpY]}\alpha^*y^*-\max_{y\in[0,y^*]}(\realc(y)+y^*-y-(1-\alpha^*)y^*(\ln y^*-\ln y))\\
=&\max_{y^*\in[0,\UpY],R^*\in[0,y^*]}y^*-R^*-\max_{y\in[0,y^*]}(\realc(y)+y^*-y-R^*(\ln y^*-\ln y))\\
=&\max_{y^*\in[0,\UpY],R^*\in[0,y^*]}\min_{y\in[0,y^*]}y-\realc(y)+R^*(\ln y^*-\ln y-1)\\
=&\max_{y^*\in[0,\UpY],R^*\in[0,y^*]}\min_{y\in[0,y^*]}\reals(y)+R^*(\ln y^*-\ln y-1).
\end{align*}
For technical convenience, we relax the constraint of $R^*\in[0, y^*]$ to $R^*\in[0,\UpY]$. This does not influence the optimal solution, as stated in the following lemma.
\begin{lemma}
\label{lem:reparameter-R}
$\max_{A\subseteq F}\Uver(A)=\max_{y^*\in[0,\UpY],R^*\in[0,\UpY]}\min_{y\in[0,y^*]}\reals(y)+R^*(\ln y^*-\ln y-1)$.
\end{lemma}
The reparameterization makes the optimization objective jointly concave in $y^*$ and $R^*$, enabling us to reduce the agent's problem to a two-variable convex optimization.
\begin{theorem}\label{thm:tractable-convex-program}
$\Phi(y^*,R^*):=\min_{y\in[0,y^*]}\reals(y)+R^*(\ln y^*-\ln y-1)$ is jointly concave in $y^*\in[0,\UpY]$ and $R^*\in[0,\UpY]$. Consequently, given oracle access to $\reals(\cdot)$, the agent's problem $\max_{y^*\in[0,\UpY],R^*\in[0,\UpY]}\min_{y\in[0,y^*]}\reals(y)+R^*(\ln y^*-\ln y-1)$ can be solved within additive error $\epsilon$ in polynomial time in $\frac1\epsilon$ for any $\epsilon>0$. 
\end{theorem}

\begin{proof}
Observe that by substituting $y=\lambda y^*$ in the  minimization, we have
\begin{align*}
\Phi(y^*,R^*)=\min_{\lambda\in[0,1]}\reals(\lambda y^*)-(\ln \lambda+1)R^*.
\end{align*}

For each $\lambda\in[0,1]$, $\reals(\lambda y^*)-(\ln \lambda+1)R^*$ is jointly concave in $(y^*,R^*)$. Consequently, $\min_{\lambda\in[0,1]}\reals(\lambda y^*)-(\ln \lambda+1)R^*$ is a concave function in $(y^*,R^*)$. That is, the agent's problem is reduced to maximizing a two-variable concave function over $(y^*,R^*)\in[0,\UpY]\times[0,\UpY]$. 

Given oracle access to $\reals(\cdot)$, we construct an oracle evaluating the concave function $\Phi(y^*,R^*)$.
Calculate the derivative $\frac{d}{d\lambda}(\reals(\lambda y^*)-(\ln \lambda+1)R^*)=y^*\reals'(\lambda y^*)-\frac{R^*}{\lambda}\leq y^*\reals'(\lambda y^*)\leq y^*$, where we recall $\reals'(y)=1-\realc'(y)\leq 1$ for all $y\in[0,1]$.
So for any $\lambda$ and $\lambda_1\in[\lambda,\lambda+\epsilon]$, we have
\[\reals(\lambda_1 y^*)-(\ln \lambda_1+1)R^*-(\reals(\lambda y^*)-(\ln \lambda+1)R^*)\leq\epsilon y^*\leq \epsilon.\]

Therefore, given any $(y^*,R^*)\in[0,\UpY]\times[0,\UpY]$, we can evaluate $\Phi(y^*,R^*)$ within $\epsilon$ error by enumerating an $\epsilon$-grid for $\lambda\in[0,1]$ with $O(\epsilon^{-1})$ queries.

This reduces the problem to a standard convex optimization over $(y^*,R^*)\in[0,\UpY]\times[0,\UpY]$, which can be solved within $\epsilon$ error within $O(\mathrm{poly}(\log(\epsilon^{-1})))$ iterations by standard algorithms. As each iteration takes $O(\epsilon^{-1})$ time to evaluate $\Phi(y^*,R^*)$, the total running time is polynomial in $\epsilon^{-1}$.

\qed
\end{proof}
\subsection{Structural Properties}
Additionally, we present some useful structural properties that will be utilized in the welfare analysis. 

Firstly, the following proposition provides a geometrical characterization of the agent's optimal choice of $R^*$ given fixed $y^*$.
Intuitively, when $y^*$ is fixed, the optimized objective $\reals(y)+R^*(\ln y^*-\ln y-1)$ is affine in $R^*$. The maximin optimization over $R$ and $y$ suggests a dual interpretation: the term $R(\ln y^*-\ln y-1)$ penalizes deviation of $\ln y$ from $\ln y^*-1$, while $R$ acts as a dual variable. This motivates a substitution $t=\ln y$, which transforms the surplus curve $\reals(y)$ to $\reals(e^t)$ over $t$. The maximin optimization becomes finding a tight supporting line of the transformed curve at $t=\ln y^*-1$.

\begin{proposition}
\label{thm:fix-y*-lowerconvex-opt}
Fix any $y^*\in[0,\UpY]$, it holds that
\[\max_{R^*\in[0,\UpY]}\min_{y\in[0,y^*]}\reals(y)+R^*(\ln y^*-\ln y-1)=\lowconv[\mathcal{H}_{y^*}](\ln y^*-1),\]
where $\mathcal{H}_{y^*}:=\{(t,\reals(e^t)):t\in(-\infty,\ln y^*]\}$ is a curve transformed from the graph of $\reals(y)$ on interval $y\in(0,y^*]$, substituting $t=\ln y$.

Specifically, any $R^*\in[0,\UpY]$ is an optimal solution  if and only if $R^*$ is a subgradient of $\lowconv[\mathcal{H}_{y^*}](\cdot)$ at $\ln y^*-1$.
\end{proposition}

Next, we show another property that the agent's optimal choice of action reward $y^*$ is at least the smallest surplus-maximizing action reward. This property can be useful to narrow the possible range of optimal $y^*$. For example, if the feasible surplus curve $\reals(y)$ is strictly increasing on $[0,1]$, the agent's optimal choice of $y^*$ is always $1$. 
\begin{proposition}
\label{lem:position-y*>=y+}
Let $y^+\in\arg\max_{y\in[0,\UpY]}\reals(y)$ be the smallest action reward achieving the maximal feasible surplus.
For any optimal choice of $y^*$ for the agent, it holds that $y^*\geq y^+$.
It follows that
\begin{align*}
\max_{A\subseteq F}\Uver(A)=\max_{y^*\in[y^+,\UpY]}\lowconv[\mathcal{H}_{y^*}](\ln y^*-1).
\end{align*}
\end{proposition}



To help understand the characterization results, we present the following example.
\begin{example}
\label{example:lowerconvexhull-optR}
Consider an instance with feasible surplus curve
$\reals(y)=\theta y$ 
for $y\in[0,1]$, and $\theta\in(0,1)$ is a constant.
By \Cref{lem:position-y*>=y+}, since $\reals(y)$ is strictly increasing in $y\in[0,1]$ and the maximal feasible surplus is achieved at $y^+=1$, the agent's optimal choice of $y^*$ is $y^*=1$.
In \Cref{thm:fix-y*-lowerconvex-opt}, the point set $\mathcal{H}_{y^*}=\{(t,\reals(e^t)):t\in(-\infty,\ln y^*]\}$ for $y^*=1$ is given by
\[\mathcal{H}_{1}:=\{(t,\theta e^t):t\in(-\infty,0]\}.\]
Since $\theta e^t$ is convex in $t$, the lower convex hull of $\mathcal{H}_1$ is itself, i.e., 
\[\lowconv[\mathcal{H}_{1}](t)=\theta e^t.\]
By \Cref{thm:fix-y*-lowerconvex-opt}, we have that the agent's optimal utility is 
\[\lowconv[\mathcal{H}_{y^*}](\ln y^*-1)=\lowconv[\mathcal{H}_{1}](-1)=\frac{\theta}{e},\]
and the optimal choice of $R^*=(1-\alpha^*)y^*$ (i.e. the principal's revenue) is the subgradient of $\lowconv[\mathcal{H}_{y^*}]$ at $\ln y^*-1$, which equals
\[\left.\frac{d \theta e^t}{dt}\right|_{t=-1}=\left.\theta e^t\right|_{t=-1}=\frac{\theta}{e}.\]

Remarkably, observe that the maximum feasible surplus in the example is $S^*:=\reals(y^+)=\theta$, and the agent's utility $\frac{\theta}{e}$ is a $\frac1e$ fraction of $S^*$. We will prove a matching lower bound in the next section: the agent's utility under optimal disclosed action space always achieves at least $\frac1e S^*$.
\end{example}

\section{Utility and Welfare Implications under Verifiable Costs}
\label{sec:welfare}
In this section, we analyze the welfare implications of strategic disclosure under verifiable costs. We mainly focus on the binary-outcome setting (or when the principal is restricted to linear contracts), while our guarantee on agent's utility can be extended to general outcome spaces.
We consider two types of baselines: The first-best surplus $S^*$, and the utilities in the canonical model with public knowledge of $F$.

\newcommand{\Ucan}{U^{\mathrm{can}}}
\newcommand{\Rcan}{R^{\mathrm{can}}}
\newcommand{\Scan}{S^{\mathrm{can}}}

We denote the principal's revenue, the agent's utility, and the social welfare in the canonical model by $\Rcan$, $\Ucan$, $\Scan$, respectively. Formally, let $w^{\mathrm{can}}: \mathcal{O} \to [0, \infty)$ be the revenue-optimal contract for feasible action space $F$ in the canonical model, i.e.,
$w^{\mathrm{can}}\in\arg\max_{w:\mathcal{O} \to [0, \infty)}\mathrm{Rev}(w,F)$,
breaking ties in favor of the agent's utility.
We define
\[\Rcan=\mathrm{Rev}(w^{\mathrm{can}},F),\quad\Ucan=\max_{a\in F}\mathbb{E}_{o \sim q_a}[w^{\mathrm{can}}(o)] - c_a,\quad\Scan=\Rcan+\Ucan.\]

\newcommand{\Udis}{U^{\mathrm{dis}}}
\newcommand{\Rdis}{R^{\mathrm{dis}}}
For convenience, we also denote the agent's utility and the principal's revenue induced by the agent's optimal disclosure strategy (under verifiable costs) by $\Udis$ and $\Rdis$.

In the remainder of this section, as our main focus is on the binary-outcome setting, we retain the notation from \Cref{sec:verifiable-cost} unless otherwise specified, and the first-best surplus can be rewritten as 
$S^*=\max_{y\in[0,\UpY]}\reals(y)$.

\subsection{Utility Guarantees Relative to First-Best Surplus}
We firstly show that under verifiable costs and binary outcome space, when the agent strategically selects the disclosed action set, her optimal utility is at least a $\frac1{e}$ fraction of $S^*$. As the social welfare is at least the agent's utility, an immediate implication is that the social welfare is approximately optimal under the agent's strategic disclosure.

\begin{theorem}
\label{thm:Uver-1/e-approximate-FBwelfare}
Under verifiable costs and binary outcome space, the agent's optimal utility $\Udis=\max_{A\subseteq F}\Uver(A)$ satisfies that
\[\Udis\geq\frac1e S^*.\]
\end{theorem}


\begin{proof}
By \Cref{lem:reparameter-R} and \Cref{thm:fix-y*-lowerconvex-opt}, we have 
\begin{align*}
\max_{A\subseteq F}\Uver(A)&=\max_{y^*\in[0,\UpY]}\lowconv[\mathcal{H}_{y^*}](\ln y^*-1)\\
&\geq \lowconv[\mathcal{H}_{y^+}](\ln y^+ -1).
\end{align*}
where $\mathcal{H}_{y^*}=\{(t,\reals(e^t)):t\in(-\infty,\ln y^*]\}$, and $y^+\in\arg\max_{y\in[0,\UpY]}\reals(y)$.
We prove that the agent's maximum utility when fixing $y^*=y^+$ suffices to achieve $\frac1e S^*$, i.e., $\lowconv[\mathcal{H}_{y^+}](\ln y^+ -1)\geq\frac1e S^*=\frac1e \reals(y^+)$.

Denote $t^\dag=\ln y^+-1$. In the lower convex hull of $\mathcal{H}_{y^+}=\{(t,\reals(e^t)):t\in(-\infty,\ln y^+]\}$, the point at $t^\dag$ is a convex combination of at most two points $(t_1,\reals(e^{t_1})),(t_2,\reals(e^{t_2}))\in\mathcal{H}_{y^+}$. That is, there exist $p_1\in[0,1],p_2=1-p_1$ such that $p_1t_1+p_2t_2=t^\dag$, and
\[p_1\reals(e^{t_1})+p_2\reals(e^{t_2})=\lowconv[\mathcal{H}_{y^+}](t^\dag).\]
By the concavity of $\reals(\cdot)$ and $\reals(0)=0$, we have
\[
\reals(y)\geq \frac{y}{y^+}\reals(y^+),\quad\forall y\in[0,y^+].
\]
Since $t_1,t_2\in(-\infty,\ln y^+]$, we have $\reals(e^{t_1})\geq \frac{e^{t_1}}{y^+}\reals(y^+)$ and $\reals(e^{t_2})\geq \frac{e^{t_2}}{y^+}\reals(y^+)$. 
Therefore, it holds that
\begin{align*}
p_1\reals(e^{t_1})+p_2\reals(e^{t_2})
\geq \frac{p_1e^{t_1}+p_2e^{t_2}}{y^+}\reals(y^+)
\geq \frac{e^{p_1t_1+p_2t_2}}{y^+}\reals(y^+)
=\frac{e^{t^\dag}}{y^+}\reals(y^+)
\end{align*}
Here the second inequality is by Jensen's inequality.
Recall that $t^\dag=\ln y^+-1$ and $\lowconv[\mathcal{H}_{y^+}](t^\dag)=p_1\reals(e^{t_1})+p_2\reals(e^{t_2})$, we obtain
\[\lowconv[\mathcal{H}_{y^+}](\ln y^+ -1)\geq \frac1e \reals(y^+)=\frac1e S^*.\]
It follows that $\Udis=\max_{A\subseteq F}\Uver(A)\geq\frac1e S^*$.
\qed
\end{proof}
This lower bound on agent's utility is tight: Recall \Cref{example:lowerconvexhull-optR}, where the feasible surplus curve is in the form of $\reals(y)=\theta y$ and $S^*=\theta$. The agent's optimal utility in this example is $\frac{\theta}{e}$, which equals $\frac1e S^*$.

Although the agent's strategic selection of disclosed action set exhibits a strong guarantee on social welfare, we show that the principal's revenue can be arbitrarily small compared to the first-best welfare.

\begin{proposition}
\label{thm:principal-rev-inapproximate-welfare}
For any $\epsilon>0$, there exists an instance with verifiable costs and binary outcome space, such that the principal's revenue $\Rdis$ satisfies
\[\Rdis=(1-\alpha^*)y^*<\epsilon\cdot S^*,\]
where the agent's optimal disclosed action set is characterized by $\alpha^*$ and $y^*$ as in \eqref{eq:Uver-with-MinCost}, inducing the principal's optimal contract $w_{\alpha^*}$ and anticipated action reward $y^*$.
\end{proposition}
The proof of this proposition is by analyzing an instance with the feasible surplus curve
\[
\reals(y)=\begin{cases}
y,&y\in[0,\frac1M],\\
\frac1{M}(\ln(My)+1),&y\in[\frac1M,1].
\end{cases},
\]
where $M$ is a sufficiently large constant, resulting in $\frac{\Rdis}{S^*}=\frac1{1+\ln M}$. The detailed proof is in the appendix.

Despite the general impossibility, we prove a lower bound on principal's revenue when the ratio between the minimum and maximum rewards among non-null base actions is not too small.

\begin{theorem}\label{thm:L-bounded-revenueguarantee}

Define $L=\frac{\sup_{a\in F_B}y_a}{\inf_{a\in F_B\setminus\{a_0\}}y_a}$ as the ratio between the maximal reward and the minimal reward among feasible base actions except the null action $a_0$.
Under verifiable costs and binary outcome space, for any instance with bounded $L\in[1,+\infty)$, it holds that
\[\frac{\Rdis}{\Udis}\geq \frac{1-\frac1e}{\ln(L)+1}.\]

This implies that
\[\frac{\Rdis}{S^*}\geq \frac{\frac{1}e-\frac{1}{e^2}}{\ln(L)+1}=\Theta\left(\frac1{\log(L)}\right).\]
\end{theorem}
This bound is asymptotically tight: Recall the example in the proof of \Cref{thm:principal-rev-inapproximate-welfare}, where the surplus curve satisfies $L=M$, while the ratio is $\frac{\Rdis}{S^*}=\frac1{1+\ln M}$.


The logarithmic dependence on the reward spread parallels approximation guarantees for linear contracts in the canonical model \cite{SimplevsOptimalDuttingRT19}. Recall that this revenue guarantee holds for the general-outcome setting when the principal is restricted to linear contracts. Thus, even under strategic disclosure of the action space, linear contracts are still approximately optimal for the principal under mild conditions. This extends the existing arguments on the robustness of linear contracts.

\subsection{Utility Impacts Relative to Canonical Model}
Next we compare the agent's utility, principal's revenue, and the social welfare under strategic disclosure with their counterparts in the canonical model, where the true feasible set $F$ is public knowledge.

For the agent, the utility $\Udis$ under strategic disclosure is weakly higher than her utility $\Ucan$ in the canonical model. Moreover, the improvement can be unbounded.
\begin{proposition}\label{thm:impact-agentutility}
Under verifiable costs and binary outcome space, it holds for any instance that $\Udis\geq \Ucan$. Moreover, for any $K>1$ there exists an instance such that
\[\Udis>K\cdot \Ucan.\]
\end{proposition}

For the principal, the revenue $\Rdis$ is never improved under the agent's optimal strategic disclosure. Moreover, $\Rdis$ can be arbitrarily small compared with $\Rcan$.

\begin{proposition}\label{thm:impact-principalrevenue}
Under verifiable costs and binary outcome space, it holds for any instance that 
\[\Rdis\leq \Rcan.\] Moreover, for any $\epsilon>0$ there exists an instance such that
\[\Rdis<\epsilon\cdot \Rcan.\]
\end{proposition}
Strategic disclosure can either damage or improve social welfare relative to the canonical model.
On the one hand, the social welfare may be damaged by strategic disclosure in some instances: Recall \Cref{example:lowerconvexhull-optR}, where the feasible surplus curve is in the form of $\reals(y)=\theta y$ for $y\in[0,1]$, which leads to $\Udis=\Rdis=\frac1e\theta$. It is not hard to verify that in the canonical model, the principal obtains $\Rcan=\theta$ by setting $\alpha=1-\theta$, with $\Ucan=0$. As $\Udis+\Rdis=\frac{2}{e}(\Rcan+\Ucan)$, there is a welfare loss caused by strategic disclosure.
On the other hand, there also exist instances where the social welfare is improved by strategic disclosure, and the improvement can be unbounded.

\begin{proposition}\label{thm:impact-welfare}
Under verifiable costs and binary outcome space, for any $K>0$ there exists an instance such that
\[\Udis+\Rdis>K(\Ucan+\Rcan).\]
\end{proposition}

\subsection{Extension to General Outcome Spaces}
Lastly, we consider the verifiable cost setting with general outcome spaces, where the principal is not restricted to linear contracts. By extending \Cref{thm:Uver-1/e-approximate-FBwelfare}, we show that the agent still has a strategy to guarantee a $1/e$ fraction of the first-best surplus. The agent still significantly benefits from the power of strategic disclosure.
\begin{theorem}
\label{thm:Extension-GeneralOutcome-1/eS*}
Under verifiable costs and general outcome space, the agent's optimal utility $\Udis=\max_{A\subseteq F}\Uver(A)$ satisfies that
\[\Udis\geq \frac1e S^*.\]
\end{theorem}
Intuitively, the agent can construct a ``linear" action space based on a single surplus-maximizing action, so that the principal's contract is essentially linear. This enables the agent to apply the construction of the optimal cost curve as if in the binary-outcome setting. 

\section{Conclusion}
\label{sec:conclusion}
This paper introduces strategic action space disclosure as a form of endogenously generated information asymmetry in principal-agent contracting. We characterize strategic outcomes under two variants of cost-verifiability regimes, revealing that agents can guarantee a substantial fraction of first-best social surplus as their utility through strategic concealment of capabilities. Under verifiable costs and binary outcome space, we analyze the utility and  welfare implications of strategic disclosure, relative to the first-best surplus and the canonical model.




%
\begin{credits}
\subsubsection{\ackname} This research was supported by the National Natural Science Foundation of China (NSFC) under Grant No. 62572010.

We thank the anonymous reviewers of EC 2026 and WINE 2026 for their insightful comments and constructive suggestions.
\end{credits}

\bibliographystyle{ACM-Reference-Format}
\bibliography{ref.bib}

\newpage
\appendix

\section{Missing Proofs}
\subsection{Proof of \Cref{lem:reparameter-R}}
\begin{proof}
We know that $\max_{A\subseteq F}\Uver(A)=\max_{y^*\in[0,\UpY],R^*\in[0,y^*]}\min_{y\in[0,y^*]}\reals(y)+R^*(\ln y^*-\ln y-1)$, so it suffices to prove that given any fixed $y^*\in[0,\UpY]$, $\min_{y\in[0,y^*]}y-\realc(y)+R(\ln y^*-\ln y-1)$ is decreasing in $R$ for $R\in[y^*,\UpY]$, which will directly imply that
\[\max_{R^*\in[y^*,\UpY]}\min_{y\in[0,y^*]}\reals(y)+R^*(\ln y^*-\ln y-1)\leq \max_{R^*\in[0,y^*]}\min_{y\in[0,y^*]}\reals(y)+R^*(\ln y^*-\ln y-1).\]

For any $R'\in[y^*,\UpY]$, for any $y\in[0,y^*]$, since $R'\geq y^*\geq y$, observe that
\[\frac{d}{dy}(y+R'(\ln y^*-\ln y-1))=1-\frac{R'}{y}\leq 0,\]
i.e., $y+R'(\ln y^*-\ln y-1)$ is non-increasing in $y$.
Since $-\realc(y)$ is also non-increasing in $y$, $y-\realc(y)+R'(\ln y^*-\ln y-1)$ is minimized when $y=y^*$. That is, $\min_{y\in[0,y^*]}y-\realc(y)+R'(\ln y^*-\ln y-1)=y^*-\realc(y^*)+R'(\ln y^*-\ln y^*-1)=y^*-\realc(y^*)-R'$.
Therefore, $\min_{y\in[0,y^*]}y-\realc(y)+R'(\ln y^*-\ln y-1)$ is decreasing in $R'$ for $R'\in[y^*,\UpY]$. This completes the proof.
\qed
\end{proof}

\subsection{Proof of \Cref{thm:fix-y*-lowerconvex-opt}}

\begin{proof}

Substituting $t=\ln y$, we rewrite
\[\min_{y\in[0,y^*]}\reals(y)+R^*(\ln y^*-\ln y-1)=\min_{t\in(-\infty,\ln y^*]}\reals(e^t)-R^*(t-(\ln y^*-1)).\]

Denote $t^\dag=\ln y^*-1$. In the lower convex hull of $\mathcal{H}_{y^*}=\{(t,\reals(e^t)):t\in(-\infty,\ln y^*]\}$, the point at $t^\dag$ is a convex combination of at most two points $(t_1,\reals(e^{t_1})),(t_2,\reals(e^{t_2}))\in\mathcal{H}_{y^*}$. That is, there exist $p_1\in[0,1],p_2=1-p_1$ such that $p_1t_1+p_2t_2=t^\dag$, and
\[p_1\reals(e^{t_1})+p_2\reals(e^{t_2})=\lowconv[\mathcal{H}_{y^*}](t^\dag).\]

For any $R^*\in[0,\UpY]$, we have
\begin{align*}
&\min_{t\in(-\infty,\ln y^*]}\reals(e^t)-R^*(t-(\ln y^*-1))\\
\leq&\min_{t\in\{t_1,t_2\}}\reals(e^t)-R^*(t-(\ln y^*-1))\\
\leq&p_1(\reals(e^{t_1})-R^*(t_1-(\ln y^*-1)))+p_2(\reals(e^{t_2})-R^*(t_2-(\ln y^*-1)))\\
=&\lowconv[\mathcal{H}_{y^*}](t^\dag)-R^*(t^\dag-(\ln y^*-1))\\
=&\lowconv[\mathcal{H}_{y^*}](\ln y^*-1).
\end{align*}
That is, 
\[\max_{R^*\in[0,\UpY]}\min_{y\in[0,y^*]}\reals(y)+R^*(\ln y^*-\ln y-1)\leq\lowconv[\mathcal{H}_{y^*}](\ln y^*-1).\]

To achieve this upper bound, take $R^\dag$ as any subgradient of $\lowconv[\mathcal{H}_{y^*}](\cdot)$ at $t^\dag$, which exists due to convexity. Then it holds for all $t\in(-\infty,\ln y^*]$ that 
\[\reals(e^t)-\lowconv[\mathcal{H}_{y^*}](t^\dag)\geq R^\dag(t-t^\dag),\]
which implies that $\reals(e^t)-R^\dag t\geq \lowconv[\mathcal{H}_{y^*}](t^\dag)-R^\dag t^\dag$.
Then we have
\begin{align*}
&\min_{t\in(-\infty,\ln y^*]}\reals(e^t)-R^\dag(t-(\ln y^*-1))\\
\geq&\lowconv[\mathcal{H}_{y^*}](t^\dag)-R^\dag(t^\dag-(\ln y^*-1))\\
=&\lowconv[\mathcal{H}_{y^*}](\ln y^*-1).
\end{align*}
If $R^\dag\in[0,\UpY]$, we obtain the achievability by setting $R^*=R^\dag$. If $R^\dag>\UpY$, by the argument in the proof of \Cref{lem:reparameter-R}, we know $\min_{y\in[0,y^*]}\reals(y)+R(\ln y^*-\ln y-1)$ is decreasing in $R\in[\UpY,+\infty)$, so $R^*=\UpY$ also achieves the upper bound. If $R^\dag<0$, we have $\lim_{t\to-\infty}\reals(e^t)-R^\dag t=-\infty$, which contradicts $\reals(e^t)-R^\dag t\geq \lowconv[\mathcal{H}_{y^*}](t^\dag)-R^\dag t^\dag$, so this case is impossible. Therefore, the upper bound is always achievable by some $R^*\in[0,\UpY]$, that is,
\[\max_{R^*\in[0,\UpY]}\min_{y\in[0,y^*]}\reals(y)+R^*(\ln y^*-\ln y-1)=\lowconv[\mathcal{H}_{y^*}](\ln y^*-1).\]

Finally, for any $R^*\in[0,\UpY]$ such that $\min_{y\in[0,y^*]}\reals(y)+R^*(\ln y^*-\ln y-1)=\lowconv[\mathcal{H}_{y^*}](\ln y^*-1)$, we have 
$\reals(e^t)-R^*(t-(\ln y^*-1))\geq\lowconv[\mathcal{H}_{y^*}](\ln y^*-1)$
for all $t\in(-\infty,\ln y^*]$, which implies that
\[\reals(e^t)-\lowconv[\mathcal{H}_{y^*}](t^\dag)\geq R^*(t-t^\dag),\]
i.e., $R^*$ is a subgradient of $\lowconv[\mathcal{H}_{y^*}](\cdot)$ at $t^\dag=\ln y^*-1$.
\qed
\end{proof}
\subsection{Proof of \Cref{lem:position-y*>=y+}}

\begin{proof}
By the concavity of $\reals(\cdot)$, $\reals(y)$ is strictly increasing in $y\in[0,y^+)$ and weakly decreasing in $y\in[y^+,\UpY]$.

We prove that for any $y^*<y^+$, it holds that
\[\lowconv[\mathcal{H}_{y^*}](\ln y^*-1)<\lowconv[\mathcal{H}_{y^+}](\ln y^+-1).\]

By definition
$\mathcal{H}_{y^*}=\{(t,\reals(e^t)):t\in(-\infty,\ln y^*]\}$
and $\mathcal{H}_{y^+}=\{(t,\reals(e^t)):t\in(-\infty,\ln y^+]\}$, so $\mathcal{H}_{y^+}=\mathcal{H}_{y^*}\cup\{(t,\reals(e^t)):t\in(\ln y^*,\ln y^+]\}$.

Since the point $(\ln y^+-1,\lowconv[\mathcal{H}_{y^+}](\ln y^+-1))$ is a convex combination of points in $\mathcal{H}_{y^+}$, there exists a distribution $t\sim D_1$ supported in $(-\infty,\ln y^+]$ such that
\[\E_{t\sim D_1}[t]=\ln y^+-1,\text{ and }\E_{t\sim D_1}[\reals(e^t)]=\lowconv[\mathcal{H}_{y^+}](\ln y^+-1).\]

Define $D_2$ as the distribution of $t-\ln y^++\ln y^*$ when $t\sim D_1$. It holds that $\E_{t\sim D_2}[t]=\ln y^+-1-\ln y^++\ln y^*=\ln y^*-1$, and $D_2$ is supported in $(-\infty,\ln y^*]$. Consequently, by the definition of lower convex hull we have
\[\lowconv[\mathcal{H}_{y^*}](\ln y^*-1)\leq \E_{t\sim D_2}[\reals(e^t)].\]

Therefore, we have
\begin{align*}
&\lowconv[\mathcal{H}_{y^+}](\ln y^+-1)-\lowconv[\mathcal{H}_{y^*}](\ln y^*-1)\\
\geq&\E_{t\sim D_1}[\reals(e^t)]-\E_{t\sim D_2}[\reals(e^t)]\\
=&\E_{t\sim D_1}[\reals(e^t)-\reals(e^{t-\ln y^++\ln y^*})]\\
>&0.
\end{align*}
Here the strict inequality is because $\reals(y)$ is strictly increasing on $y\in[0,y^+]$. This completes the proof. \qed
\end{proof}

\subsection{Proof of \Cref{thm:principal-rev-inapproximate-welfare}}

\begin{proof}
We prove the proposition by constructing an instance represented by the feasible cost curve $\realc(\cdot)$, or equivalently the surplus curve $\reals(\cdot)$. 
Let $M>1$ be a sufficiently large constant to be determined later.

Let $\UpY=1$, and for $y\in[0,1]$, define
\[
\reals(y)=\begin{cases}
y,&y\in[0,\frac1M],\\
\frac1{M}(\ln(My)+1),&y\in[\frac1M,1].
\end{cases},
\quad\realc(y)=y-\reals(y).
\]
Since $\reals(y)$ is increasing in $y\in[0,1]$, we have that $y^+=1$ achieves the maximal surplus $S^*=\reals(y^+)=\frac1{M}(\ln(M)+1)$.
By \Cref{lem:position-y*>=y+} we have $y^*\geq y^+$, so $y^*=1$.

By \Cref{thm:fix-y*-lowerconvex-opt}, for $y^*=1$ we have
\[\max_{R^*\in[0,\UpY]}\min_{y\in[0,y^*]}\reals(y)+R^*(\ln y^*-\ln y-1)=\lowconv[\mathcal{H}_{1}](-1),\]
where $\mathcal{H}_{1}=\{(t,\reals(e^t)):t\in(-\infty,0]\}$.
Specifically, we have
\[\reals(e^t)=\begin{cases}
e^t,&t\in(-\infty,-\ln M],\\
\frac1{M}(t+\ln(M)+1),&t\in[-\ln M,0].
\end{cases},\]
which is convex in $t$. Consequently, the lower convex hull of $\mathcal{H}_{1}$ is itself, i.e.,
$\lowconv[\mathcal{H}_{1}](t)=\reals(e^t)$ for $t\in(-\infty,0]$.
Moreover, since $\reals(e^t)=\frac1{M}(t+\ln(M)+1)$ is linear on $[-\ln M,0]$, its unique subgradient at $t=-1$ is $\frac1{M}$. By \Cref{thm:fix-y*-lowerconvex-opt}, the optimal value of $R^*$ is uniquely $\frac1{M}$.
That is, under the agent's optimal strategy, the principal's revenue is $(1-\alpha^*)y^*=R^*=\frac1{M}$.

Taking $M>e^\frac{1}{\epsilon}$, we have
\[\frac{(1-\alpha^*)y^*}{S^*}=\frac{1/M}{(\ln(M)+1)/M}=\frac1{\ln(M)+1}<\epsilon.\]
\end{proof}

\subsection{Proof of \Cref{thm:L-bounded-revenueguarantee}}

\begin{proof}
Without loss of generality, we assume the action rewards are normalized to $\sup_{a\in F_B}y_a=1$. When $L<+\infty$, we have $\inf_{a\in F_B\setminus\{a_0\}}y_a=\frac1L$. This implies that $\reals(y)$ is linear in the interval $y\in[0,\frac1L]$.

We denote $y_1=\frac1L$, $t_1=\ln y_1$, and  $s_1=\reals(y_1)$. Then we have
\[\reals(y)=\frac{s_1}{y_1}y,\quad\forall y\in[0,y_1].\]

Let $\alpha^*$, $y^*$ be the optimal solution of the agent's problem \eqref{eq:Uver-with-MinCost}, implying the optimal disclosed action set $A^*$. By \Cref{thm:fix-y*-lowerconvex-opt}, we know 
that the agent's optimal utility $\Uver(A^*)$ is
\[\Uver(A^*)=\lowconv[\mathcal{H}_{y^*}](\ln y^*-1),\]
where $\mathcal{H}_{y^*}=\{(t,\reals(e^t)):t\in(-\infty,\ln y^*]\}$, and that the principal's revenue $R^*=(1-\alpha^*)y^*$ is a subgradient of $\lowconv[\mathcal{H}_{y^*}](\cdot)$ at $t^\dag:=\ln y^*-1$.

In $\mathcal{H}_{y^*}$, the linear piece of $\reals(y)$ on $y\in(0,\frac1L]$ is transformed to $\{(t,\frac{s_1}{y_1}e^t):t\in(-\infty,t_1]\}$. 
Since $R^*$ is a subgradient of $\lowconv[\mathcal{H}_{y^*}](\cdot)$ at $t^\dag$, it holds for all $t\in(-\infty,\ln y^*]$ that
\[R^*(t-t^\dag)+\lowconv[\mathcal{H}_{y^*}](t^\dag)\leq \lowconv[\mathcal{H}_{y^*}](t).\]
Since $\lowconv[\mathcal{H}_{y^*}](t)\leq \reals(e^t)\leq \frac{s_1}{y_1}e^t$ and $\lowconv[\mathcal{H}_{y^*}](t^\dag)=\Uver(A^*)$, this implies
\[\frac{s_1}{y_1}e^t\geq R^*(t-t^\dag)+\Uver(A^*).\]

We consider two cases: $R^*\leq s_1$ or $R^*>s_1$.

If $R^*\leq s_1$, take $t_2=\ln(\frac{y_1}{s_1}R^*)\leq t_1$, we have $R^*=\frac{s_1}{y_1}e^{t_2}\geq R^*(t_2-t^\dag)+\Uver(A^*)$, which implies that $(1-t_2+t^\dag)R^*\geq \Uver(A^*)$.

Observe that $s_1=\frac{s_1}{y_1}e^{t_1}=e^{t_1-t_2}R^*$.
By \Cref{thm:Uver-1/e-approximate-FBwelfare}, we have $\Uver(A^*)\geq \frac1e S^*$, which implies that $\frac1e\cdot s_1=\frac1e\cdot \reals(y_1)\leq \frac1e\cdot S^*\leq \Uver(A^*)$.
Therefore, it holds that
\[e^{t_1-t_2-1}R^*\leq  \Uver(A^*).\]
Combining with $(1-t_2+t^\dag)R^*\geq \Uver(A^*)$, we have
\[e^{t_1-t_2-1}\leq 1-t_2+t^\dag.\]
Define $\delta=t_1-t_2-1$, we have $e^{\delta}\leq\delta+t^\dag-t_1+2$. It follows that $\delta\leq \ln(\delta+t^\dag-t_1+2)\leq \frac1e(\delta+t^\dag-t_1+2)$, i.e., $\delta\leq \frac1{e-1}(t^\dag-t_1+2)$. Recall that $(1-t_2+t^\dag)R^*\geq \Uver(A^*)$, we have
\[R^*\geq \frac{\Uver(A^*)}{1+t^\dag-(t_1-1-\delta)}= \frac{\Uver(A^*)}{2+t^\dag-t_1+\delta}\geq \frac{\Uver(A^*)}{2+t^\dag-t_1+\frac1{e-1}(t^\dag-t_1+2)}=\frac{\Uver(A^*)}{\frac{e}{e-1}(t^\dag-t_1+2)}.\]

If $R^*>s_1$, we have
$R^*>s_1=\frac{s_1}{y_1}e^{t_1}\geq R^*(t_1-t^\dag)+\Uver(A^*)$, which implies that
$(1-t_1+t^\dag)R^*\geq\Uver(A^*)$. 

Combining the two cases, since $\frac{e}{e-1}(t^\dag-t_1+2)>(1-t_1+t^\dag)$, it holds that
\[\frac{\Rdis}{\Udis}=\frac{R^*}{\Uver(A^*)}\geq \frac1{\frac{e}{e-1}(t^\dag-t_1+2)}=\frac1{\frac{e}{e-1}(\ln y^*-t_1+1)}\geq\frac1{\frac{e}{e-1}(\ln(L)+1)}.\]


By \Cref{thm:Uver-1/e-approximate-FBwelfare} we have $\Uver(A^*)\geq \frac1eS^*$, which implies 
\[\frac{\Rdis}{S^*}=\frac{R^*}{S^*}\geq\frac{\frac{e-1}{e^2}}{\ln(L)+1}.\]
\qed
\end{proof}

\subsection{Proof of \Cref{thm:impact-agentutility}}
\begin{proof}
Observe that $\Ucan=\Uver(F)$, i.e., the agent gets the original utility if she discloses $A=F$. It follows immediately that $\Udis=\max_{A\subseteq F}\Uver(A)\geq \Ucan$.

To show the existence of an instance where $\Udis>K\cdot \Ucan$, it suffices to recognize that there exists an instance with $\Ucan=0$ and $\Udis>0$. Recall \Cref{example:lowerconvexhull-optR}, where the feasible surplus curve is in the form of $\reals(y)=\theta y$ for $y\in[0,1]$, which leads to $\Udis=\Rdis=\frac1e\theta$. In the canonical model, since the cost curve is $\realc(y)=y-\reals(y)=(1-\theta)y$, the principal's optimal contract incentivizes the action reward $1$ by setting $\alpha=1-\theta$, inducing $\Rcan=\theta$ and $\Ucan=0$. This completes the proof.
\qed
\end{proof}

\subsection{Proof of \Cref{thm:impact-principalrevenue}}

\begin{proof}
We firstly prove that $\Rdis\leq\Rcan$ always holds under verifiable costs and binary outcome space.
Consider any instance represented by the feasible surplus curve $\reals(\cdot)$.
Let $\alpha^*,y^*$ represent the agent's optimal strategy as in \eqref{eq:Uver-with-MinCost}, inducing revenue $\Rdis=R^*=(1-\alpha^*)y^*$ for the principal. By \Cref{thm:fix-y*-lowerconvex-opt}, $R^*$ is a subgradient of $\lowconv[\mathcal{H}_{y^*}](\cdot)$ at $\ln y^*-1$. By the definition of a subgradient, we have $\lowconv[\mathcal{H}_{y^*}](\ln y^*)-\lowconv[\mathcal{H}_{y^*}](\ln y^*-1)\geq R^*(\ln y^*-(\ln y^*-1))=R^*$.

Recall that $\mathcal{H}_{y^*}=\{(t,\reals(e^t)):t\in(-\infty,\ln y^*]\}$. By the definition of the lower convex hull, we have that $\lowconv[\mathcal{H}_{y^*}](\ln y^*)=\reals(y^*)$, and that there exists a distribution $t\sim D_1$ supported on $(-\infty,\ln y^*]$ such that $\E_{t\sim D_1}[t]=\ln y^*-1$ and $\E_{t\sim D_1}[\reals(e^t)]=\lowconv[\mathcal{H}_{y^*}](\ln y^*-1)$. 
It follows that \[\E_{t\sim D_1}[\reals(y^*)-\reals(e^t)]=\reals(y^*)-\E_{t\sim D_1}[\reals(e^t)]=\lowconv[\mathcal{H}_{y^*}](\ln y^*)-\lowconv[\mathcal{H}_{y^*}](\ln y^*-1)\geq R^*.\]

In the canonical model, given the cost curve $\realc(y)=y-\reals(y)$, the principal can induce any action reward $y\in[0,\UpY]$ by the linear contract with $\alpha=\realc'_-(y)=1-\reals'_-(y)$, obtaining a revenue of $(1-\alpha)y=\reals'_-(y)\cdot y$. Therefore, we have $\Rcan=\max_{y\in[0,\UpY]}\reals'_-(y)\cdot y$. It follows that for any $y\in[0,y^*]$, we have $\reals'_-(y)\cdot y\leq \Rcan$.

Observe that for any $t\leq\ln y^*$, $\reals(y^*)-\reals(e^t)=\int_t^{\ln y^*}\reals'_-(e^z)e^zdz\leq \Rcan(\ln y^*-t)$, where the inequality is because $\reals'_-(e^z)\cdot e^z\leq \Rcan$ for all $z\leq \ln y^*$. It follows that
\[\Rdis=R^*\leq \E_{t\sim D_1}[\reals(y^*)-\reals(e^t)]\leq \E_{t\sim D_1}[\Rcan(\ln y^*-t)]=\Rcan(\ln y^*-(\ln y^*-1))=\Rcan.\]

Next, given any $\epsilon>0$, we construct an instance with $\Rdis<\epsilon\cdot \Rcan$.
Let $M>e$ be a constant to be determined later. Consider the instance with the feasible surplus curve
\[
\reals(y)=\min\{y,\frac1M\}.
\]
We have $\reals'_-(y)=\I[y\leq \frac1M]$ and $\Rcan=\max_{y\in[0,1]}\reals'_-(y)\cdot y=\frac1M$ in the canonical model.

With strategic disclosure, let the agent's optimal utility be represented by $y^*$ and $R^*$, i.e. the induced action reward and principal's revenue.
By \Cref{thm:fix-y*-lowerconvex-opt}, $R^*$ is a subgradient of $\lowconv[\mathcal{H}_{y^*}](\cdot)$ at $\ln y^*-1$. Moreover, we have $y^*\in\arg\max_{y\in[0,1]}\lowconv[\mathcal{H}_{y}](\ln y-1)$.

We show that $y^*=1$ in this instance. By \Cref{lem:position-y*>=y+} we have $y^*\geq y^+=\frac1M$, so it suffices to prove that $\lowconv[\mathcal{H}_{y}](\ln y-1)$ is strictly increasing on $y\in[\frac1M,1]$. Recall that $\mathcal{H}_{y}=\{(t,\reals(e^t)):t\in(-\infty,\ln y]\}$, where $\reals(e^t)=\min\{e^t,\frac1M\}$.
Since $\reals(e^t)$ is convex on $(-\infty,-\ln M]$ and is constant on $[-\ln M,\ln y]$, there exists some $t_y\leq -\ln M$ such that $\lowconv[\mathcal{H}_{y}](t)=\begin{cases}
    e^t,&t\leq t_y,\\
    e^{t_y}(1+t-t_y),&t\in[t_y,\ln y]
\end{cases}$ and that $e^{t_y}(\ln y-t_y)=\frac1{M}-e^{t_y}$, that is, the lower convex hull ``irons" the interval $[t_y,\ln y]$. Rewrite $\ln y=t_y+(\frac1{Me^{t_y}}-1)$. Since $\frac{d}{d t_y}(t_y+(\frac1{Me^{t_y}}-1))=1-\frac1{Me^{t_y}}\leq 0$ for any $t_y\leq -\ln M$, with the inequality strictly holding for $t_y<-\ln M$, it holds that $t_y$ is decreasing in $y$. 
Since $\lowconv[\mathcal{H}_{y}](\ln y-1)=\frac1{M}-e^{t_y}$ when $t_y\leq \ln y-1$, and $\lowconv[\mathcal{H}_{y}](\ln y-1)=\frac{y}{e}$ when $t_y\geq \ln y-1$, we have $\lowconv[\mathcal{H}_{y}](\ln y-1)$ is strictly increasing in $y\in[\frac1M,1]$.
This implies $y^*=1$.

Considering $\lowconv[\mathcal{H}_{1}](\cdot)$, there exists $t_{1}\leq -\ln M$ such that $e^{t_1}(0-t_1)=\frac1{M}-e^{t_1}$ and $\lowconv[\mathcal{H}_{1}](-1)=\frac1{M}-e^{t_1}$, while the subgradient of $\lowconv[\mathcal{H}_{1}](\cdot)$ at $-1$ is uniquely $e^{t_1}$. Then we have $R^*=e^{t_1}=\frac1{M(1-t_1)}\leq \frac1{M(1+\ln M)}$, and it follows that
\[\Rdis=R^*\leq \frac1{1+\ln M}\Rcan.\]
The statement is satisfied by taking $M\geq e^{\frac1\epsilon}$.\qed
\end{proof}

\subsection{Proof of \Cref{thm:impact-welfare}}

\begin{proof}
Recall that $\Udis\geq\frac1e S^*$ always holds by \Cref{thm:Uver-1/e-approximate-FBwelfare}. We construct an instance such that $\Rcan+\Ucan\leq \frac1{Ke}S^*$, by slightly modifying the instance in the proof of \Cref{thm:principal-rev-inapproximate-welfare}. 

Let $\delta\in(0,1)$ and $M>1$ be constants to be determined later. Consider the instance with the feasible surplus curve
\[
\reals(y)=\begin{cases}
y,&y\in[0,\frac1M],\\
\frac1{M}(\delta\ln(My)+1),&y\in[\frac1M,1].
\end{cases},
\]
and the minimum feasible cost
\[
\realc(y)=y-\reals(y)=\begin{cases}
0,&y\in[0,\frac1M],\\
y-\frac1{M}(\delta\ln(My)+1),&y\in[\frac1M,1].
\end{cases}
\]
It holds that $\realc'_-(y)=\begin{cases}
0,&y\in[0,\frac1M],\\
1-\frac{\delta}{My},&y\in(\frac1M,1].
\end{cases}$, and $S^*=\reals(1)=\frac{1+\delta\ln M}M$. 

In the canonical model, the principal can induce any action reward $y\in(\frac1M,1]$ by the linear contract with $\alpha=\realc'_-(y)=1-\frac{\delta}{My}$, or induce action reward $y=\frac1M$ by setting $\alpha=0$. In the former case, the revenue is $(1-\alpha)y=\frac{\delta}{M}$; in the latter case, the revenue is $\frac{1}{M}$. Therefore, the principal's optimal contract sets $\alpha=0$, inducing social surplus $\reals(\frac1M)=\frac1M$. That is,
\[\frac{\Udis+\Rdis}{\Ucan+\Rcan}\geq\frac{\frac1e S^*}{\frac1M}=\frac{1+\delta\ln M}e.\]
The statement is satisfied by taking any $\delta\in(0,1)$ and $M>e^{\frac{eK}{\delta}}$.
    \qed
\end{proof}
\subsection{Proof of \Cref{thm:Extension-GeneralOutcome-1/eS*}}
\begin{proof}
Assume $S^*>0$, otherwise the agent trivially obtains $0$ utility by $A=\{a_0\}$ with the null action $a_0=(q_{a_0},0)\in F_B$.

Let $a^+\in\arg\max_{a\in F_B}r(a)-c_a$ be an arbitrary feasible action achieving maximal surplus $r(a^+)-c_{a^+}=S^*$. 
For convenience, denote $r^+=r(a^+)$ and $c^+=c_{a^+}$.

Construct the disclosed action set $A$ as a continuum of actions
\[A=\{a_{(z)}=(q_{(z)},c_{(z)}):z\in[0,r^+]\},\]
where $q_{(z)}:=\frac{z}{r^+}q_{a^+}+(1-\frac{z}{r^+})q_{a_0}\in\Delta(\mathcal{O})$, and
\begin{align*}
c_{(z)}:=&\frac1e c^++\int_{\frac1e r^+}^{z}\max\{0,1-\tfrac{\frac1e(r^+-c^+)}{t}\}dt\\
=&\begin{cases}
\frac1e c^+ +z-\frac1e r^+ -\frac1e(r^+-c^+)\ln(\frac{e\cdot z}{r^+}),&z\in[\frac1e(r^+-c^+),r^+],\\
-\frac1e(r^+-c^+)\ln(\frac{r^+-c^+}{r^+}),&z\in[0,\frac1e(r^+-c^+)].\\
\end{cases}
\end{align*}
One can easily verify that $c_{(z)}$ is non-decreasing and convex, and that $c_{(z)}\geq \frac{z}{r^+}c_{a^+}$ for $z\in[0,r^+]$. Therefore $A$ is feasible as a disclosed action set, i.e., $A\subseteq F$.

Now we verify that $\Uver(A)\geq\frac1e S^*=\frac1e(r^+-c^+)$.
For any contract $w:\mathcal{O}\to[0,+\infty)$, define $\beta_1:=\E_{o\sim q_{a^+}}[w(o)]$ and $\beta_0:=\E_{o\sim q_{a_0}}[w(o)]$ as the expected payment for action $a^+$ and $a_0$, respectively. Observe that
\begin{align*}
\E_{o\sim q_{(z)}}[w(o)]=\frac{z}{r^+}\E_{o\sim q_{a^+}}[w(o)]+\left(1-\frac{z}{r^+}\right)\E_{o\sim q_{a_0}}[w(o)]=(\beta_1-\beta_0)\frac{z}{r^+}+\beta_0.
\end{align*}
Therefore, the incentivized action only depends on $(\beta_1-\beta_0)$. Since $r(a_0)=0$ and $r(a^+)>0$, outcomes $o$ with $r(o)>0$ in the support of $q_{a^+}$ never appear in the support of $q_{a_0}$. Therefore, the principal optimally sets $\beta_0=0$ without affecting the incentives.
Then it holds under the principal's optimal contract $w$ that $\E_{o\sim q_{(z)}}[w(o)]=\beta_1\frac{z}{r^+}$.

Since $c_{(z)}$ is constant on $z\in[0,\frac1e(r^+-c^+)]$, only actions $a_{(z)}$ with $z\geq \frac1e(r^+-c^+)$ may be incentivized.
The derivative of $c_{(z)}$ on $z\in[\frac1e(r^+-c^+),r^+]$ is $\frac{d c_{(z)}}{dz}=1-\frac{\frac1e(r^+-c^+)}{z}$.
Therefore, for any $z\in[\frac1e(r^+-c^+),r^+]$, the optimal contract incentivizing $a_{(z)}\in A$ sets
$\beta_1=r^+\cdot\frac{d c_{(z)}}{dz}=r^+\cdot(1-\frac{\frac1e(r^+-c^+)}{z})$, resulting in principal's revenue
\[r(a_{(z)})-\beta_1\frac{z}{r^+}=z-z(1-\frac{\frac1e(r^+-c^+)}{z})=\frac1e(r^+-c^+),\]
which is equal for all $z\in[\frac1e(r^+-c^+),r^+]$. By the assumption that the principal breaks ties in favor of the agent, the principal's optimal contract incentivizes the action $a_{(r^+)}$ by setting $\beta_1=r^+\cdot(1-\frac{\frac1e(r^+-c^+)}{r^+})=r^+-\frac1e(r^+-c^+)$, inducing payment $\beta_1\frac{r^+}{r^+}=r^+-\frac1e(r^+-c^+)$ and cost
\[c_{(r^+)}=\frac1e c^+ +r^+-\frac1e r^+ -\frac1e(r^+-c^+)\ln(\frac{e\cdot r^+}{r^+})=r^+-\frac2e(r^+-c^+).\]
Consequently, the agent's utility is
\[\Uver(A)=\beta_1\frac{r^+}{r^+}-c_{(r^+)}=\frac1e(r^+-c^+)=\frac1e S^*.\]

It follows that $\Udis\geq\Uver(A)=\frac1{e}S^*.$
\qed
\end{proof}

\end{document}